\documentclass[a4paper]{article}

\usepackage[english]{babel}
\usepackage{german}
\usepackage{amssymb,amsfonts,amsmath,amsthm,cite}
\usepackage[latin1]{inputenc}
\usepackage{fullpage}
\usepackage{graphicx}
\usepackage{subfigure}

\usepackage{float}
\usepackage{algorithm}
\usepackage[noend]{algorithmic}
\usepackage{color,hyperref}

\newcommand{\source}{{s}}
\newcommand{\sink}{{y}}
\newcommand{\Po}{{\textsc{Minimum Arrival}}}
\newcommand{\Pt}{{\textsc{Best Policy}}}

\usepackage[numbers,sort]{natbib}
\usepackage{mathtools}

\newtheorem{theorem}{Theorem}

\newtheorem{corollary}{Corollary}

\newtheorem{definition}{Definition}
\newtheorem{observation}{Observation}

\newtheorem{lemma}{Lemma}

\newtheorem{problem}{Problem}
\newtheorem{remark}{Remark}

\title{\vspace{-0.5cm}How fast can we reach a target vertex \\
in stochastic temporal graphs?\thanks{This work was supported by the NeST initiative 
of the EEE/CS School of the University of Liverpool 
and by the EPSRC grants EP/P020372/1 and EP/P02002X/1.}}

\author{Eleni C.~Akrida\thanks{Department of Computer Science, University of Liverpool, Liverpool, UK. 
Email: \texttt{e.akrida@liverpool.ac.uk}} 
\and George B.~Mertzios\thanks{Department of Computer Science, Durham University, Durham, UK. 
Email: \texttt{george.mertzios@durham.ac.uk}} 
\and Sotiris Nikoletseas\thanks{Computer Engineering \& Informatics Department, University of Patras and CTI, Greece. 
Email: \texttt{nikole@cti.gr}}
\and Christoforos Raptopoulos\thanks{Computer Engineering \& Informatics Department, University of Patras and CTI, Greece. 
Email: \texttt{raptopox@ceid.upatras.gr}}
\and Paul G.~Spirakis\thanks{Department of Computer Science, University of Liverpool, UK and 
Computer Engineering \& Informatics Department, University of Patras, Greece. 
Email: \texttt{p.spirakis@liverpool.ac.uk}}
\and Viktor Zamaraev\thanks{Department of Computer Science, Durham University, Durham, UK. 
Email: \texttt{viktor.zamaraev@durham.ac.uk}}}
\date{\vspace{-1.0cm}}

\begin{document}

\maketitle

\begin{abstract}
Temporal graphs are used to abstractly model real-life networks that are inherently dynamic in nature, 
in the sense that the network structure undergoes discrete changes over time. 
Given a static underlying graph $G=(V,E)$, a temporal graph on $G$ is a sequence of \emph{snapshots} 
$\{G_t=(V,E_t) \subseteq G: t\in \mathbb{N}\}$, one for each time step $t\geq 1$.
In this paper we study \emph{stochastic temporal graphs}, i.e.~stochastic processes $\mathcal{G}=\{G_t \subseteq G: t \in \mathbb{N}\}$ whose random variables are the snapshots of a temporal graph on~$G$. 
A natural feature of stochastic temporal graphs which can be observed in various real-life scenarios is 
a \emph{memory effect} in the appearance probabilities of particular edges; 
that is, the probability an edge $e\in E$ appears at time step $t$ depends on its 
appearance (or absence) at the previous $k$ steps. 
In this paper we study the hierarchy of models \emph{memory-$k$}, $k\geq 0$, 
which address this memory effect in an \emph{edge-centric} network evolution: 
every edge of $G$ has its own probability distribution for its appearance over time, 
\emph{independently} of all other edges. 
Clearly, for every $k\geq 1$, memory-$(k-1)$ is a special case of memory-$k$.
However, in this paper we make a clear distinction between the values $k=0$ (\emph{``no memory''}) 
and $k\geq 1$ (\emph{``some memory''}), as in some cases these models exhibit a fundamentally different 
computational behavior for these values of $k$, as our results indicate.
For every $k\geq 0$ we investigate the computational complexity of two naturally related, but fundamentally different, 
\emph{temporal path} (or \emph{journey}) problems: \Po\ and \Pt. 
In the first problem we are looking for the \emph{expected arrival time} of a foremost journey between two designated vertices~$\source,\sink$. 
In the second one we are looking for the expected arrival time of the \emph{best policy} for actually choosing 
a \emph{particular} $\source$-$\sink$ journey. 
We present a detailed investigation of the computational landscape of both problems for the different 
values of memory $k$. Among other results we prove that, surprisingly, \Po\ is \emph{strictly harder} 
than \Pt; in fact, for $k=0$, 
\Po\ is~\#P-hard while \Pt\ is solvable in $O(n^2)$ time.\newline

\end{abstract}

\section{Introduction}\label{sec:intro}

Dynamic network analysis, i.e.~analysis of networks that change over time, is currently one of the most active topics of research in network science and theory. 
A common task in this field is to use our prior knowledge of the network link dynamics to answer questions 
about the behavior of the network over time, e.g.~how quickly information can flow through it. 
Many modern real-life networks are dynamic in nature, in the sense that the network structure 
undergoes discrete changes over time \cite{michailCACM,Santoro11}. 
Here we deal with the discrete-time dynamicity of the network links (edges) over a fixed set of nodes (vertices). 
That is, given an underlying static graph $G$, the network evolution over $G$ is given by the successive appearance or absence of each edge of $G$ at every time step $t=1,2,\ldots$. 
This concept of dynamic network evolution is given by \emph{temporal graphs}~\cite{kempe,mertziosMCS13}, 
which are also known by other names such as \emph{evolving graphs} \cite{LeonardiALLM16, Ferreira-MANETS-04}, or \emph{time-varying graphs} \cite{krizanc1}. 
For a recent attempt to integrate existing models, concepts, and results from the distributed computing perspective, see~the survey papers~\cite{flocchini1,flocchini2} and the references therein.

\begin{definition}[Temporal graph]
\label{temporal-def}
Given an underlying static graph $G=(V,E)$ on~$n$ vertices and $m$ edges, a \emph{temporal graph} 
on $G$ is a sequence $\mathcal{G} = \{G_t = (V,E_t): t \in \mathbb{N}\}$ of graphs 
such that $E_t \subseteq E$ for all $t \in \mathbb{N}$. Every $G_t$ is the \emph{snapshot} 
of $\mathcal{G}$ at time step $t$. 
\end{definition}

Another way to think about temporal graphs is by assigning \emph{time-labels} on the edges; 
for example, if an edge $e$ appears in the snapshots $G_3$, $G_5$, and $G_8$, then we equivalently assign 
to $e$ the set of labels $\lambda(e)=\{3,5,8\}$. 
Due to the vast applicability of temporal graphs, various structural and algorithmic properties 
of them have been studied extensively, both via theoretical/algorithmic analysis and via empirical 
simulation-based analysis. 
In many of these works, one of the central temporal notions is that of a temporal path. 
A path in the underlying (static) graph $G$ is a \emph{temporal path} (or \emph{journey}) if there exists 
an increasing sequence of time-labels as one walks along the edges of the path~\cite{kempe,mertziosMCS13}. 
Motivated by the fact that, due to causality, information in temporal graphs can only flow along 
sequences of edges that appear in an increasing time order, many temporal graph parameters and optimization 
problems that have been studied so far are based on the notion of a temporal path and 
other related notions, e.g.~temporal analogs of distance, diameter, connectivity, reachability, 
and exploration~\cite{akridaGMS17,akridaGMS16,HenriST18,OrdaR96,BasuYBR14,BasuBRJ10,
CasteigtsFGSY15,erlebach,FlocchiniMS13,AvinKL08,LamprouMS18,enright2018deleting}.
In addition to temporal paths, recently also various temporal non-path problems have been introduced 
and algorithmically studied, such as temporal vertex cover~\cite{AkridaMSZ18}, 
temporal coloring~\cite{MertziosMZ-AAAI19}, 
and temporal $\Delta$-cliques~\cite{viardCliqueTCS,HimmelMNS17}.

Apart from the focus on the various algorithmic problems that one can study on temporal graphs, 
one can also view temporal graphs through several different levels of knowledge about the 
actual network evolution. 
On the one extreme, we may be given the whole temporal graph instance in advance, 
i.e.~the times of appearance and absence of every edge at all times, as it typically happens 
e.g.~when modeling transportation networks.
On the other extreme, the temporal graph may be created by an adversary who reveals it to us 
snapshot-by-snapshot at every time step. 
Here we focus on the intermediate knowledge settings, captured by \emph{stochastic temporal graphs}, 
where the network evolution is given by a probability distribution that governs the appearance 
of each edge over~time. 

\begin{definition}[Stochastic temporal graph]
\label{stochastic-temp-graphs-def}
A \emph{stochastic temporal graph} is a stochastic process $\mathcal{G}=\{G_t: t \in \mathbb{N}\}$ 
whose random variables are snapshots $G_t \subseteq G$ of an underlying graph~$G$. 
Every instantiation of $\mathcal{G}$ is a temporal graph. 
\end{definition}

A natural feature of stochastic temporal graphs which can be observed in various real-life scenarios 
(and which we address in this paper) is that the appearance probability of a particular edge 
at a given time step $t$ depends on the appearance (or absence) of the same edge 
at the previous $k\geq 1$ time steps. 
This ``memory effect'' can often be observed, among others, in faulty network communication and in mobile, social, 
and peer-to-peer networks~\cite{ClementiMPS11,sch,Pittel87}. 
Several other models of temporal networks which exhibit some sort of probabilistic behavior 
have been considered in the past, see e.g.~\cite{Holme-Saramaki-book-13}.

In this paper, we study a hierarchy of models for stochastic temporal graphs. These models 
concern an \emph{edge-centric} network evolution, i.e.~they assign to every edge of the underlying graph $G$ 
a probability distribution for its appearance over time, independently of all the other edges. 
The first and most basic model (\emph{memoryless} or \emph{memory-$0$}) assigns independently 
to every edge~$e$ a probability $p_e$ such that, at every time step, $e$ appears with probability 
$p_e$. In the general model (\emph{memory-$k$}), at every time step the appearance probability of every edge 
is a function of the history of its appearances/absences in the last $k\geq 1$ 
time steps. Clearly, for every $k\geq 1$, the memory-$(k-1)$ model is a special case of the memory-$k$ model.
However, in this paper we make a clear distinction between the values $k=0$ (\emph{``no memory''}) 
and $k\geq 1$ (\emph{``some memory''}), as in some cases these models exhibit a fundamentally different 
computational behavior for these values of $k$, as our results indicate (see Section~\ref{sec:best-policy}).

Our memory-$k$ model, $k\geq 1$, is a direct generalization of the homogeneous version of the memory-1 model 
that was introduced in a seminal paper by Clementi et al.~\cite{clementiMMPS10}, 
in which all edges have the same probability distribution for their appearance, 
based on their own appearance/absence at the previous step. 
In this homogeneous memory-1 model, Clementi et al.~gave upper bounds for the flooding time 
and they provided tight characterizations of the graphs on which the flooding time is constant~\cite{clementiMMPS10}. 
It is worth noting here that Avin et al.~\cite{AvinKL08} studied the completely opposite extreme 
of our edge-centric evolution; namely they considered a \emph{graph-centric} evolution model where a global 
probability distribution assigns specific transition probabilities among different snapshots~\cite{AvinKL08}. 
Between the two extremes of the edge-centric and the graph-centric network evolution models, 
there exists a whole hierarchy of locally interdependent probabilistic patterns, i.e.~probability 
distributions where the appearance probability of one edge also depends on the appearance 
of \emph{other edges} over time; such models remain mostly unexplored.

In both our memoryless and memory-$k$ variations of stochastic temporal graphs, we study two fundamental 
temporal path (i.e.~journey) problems that are defined on two designated vertices $\source$ and~$\sink$. 
Consider a piece of information that is generated at $\source$ at time 1, which we would like to send 
to $\sink$ via an $\source$-$\sink$ journey. The \emph{arrival time} of an $\source$-$\sink$ journey 
in a realization of a stochastic temporal graph is the time the information reaches $\sink$ using 
this journey. A \emph{foremost} $\source$-$\sink$ journey is one with the smallest arrival time. 
In the first part of the paper we investigate the complexity of computing the \emph{expected arrival time} of 
a \emph{foremost} $\source$-$\sink$ journey. 
Basu et al.~\cite{BasuGST12} and Nain et al.~\cite{nainTJBBY13} studied a similar problem but their work 
is restricted to the simpler cases where the underlying graph is either a path or a grid.

In the second part of the paper we investigate the complexity of computing the arrival time of a 
\emph{best policy} for actually choosing a particular $\source$-$\sink$ journey in the 
stochastic temporal graph. To illustrate this notion of a best policy, assume that some 
piece of information is carried by an entity, say Alice. Alice is given as input the parameters of 
the stochastic temporal graph (i.e.~the probabilistic rules on the edges) and, at every time step, 
she knows the current snapshot and her current location. 
Based on this information, Alice has to decide at every step for her next action, 
while her goal is to reach $\sink$ as quickly as possible on expectation, starting at time 1. 
In a very inspiring paper, Basu et al.~\cite{BasuBRJ10} consider this problem in the special case 
of the memoryless model where all edges have the same probability of appearance at every time, 
and give a Dijkstra-like polynomial-time algorithm. 
Special cases of the memory-1 model were considered in~\cite{BasuYJB14}.

To illustrate the difference between the two problems we study, we make the following analogy. 
In the first problem (\Po) we try to transfer information from $\source$ to~$\sink$ using an unbounded number of messages, i.e.~we ``flood'' the stochastic temporal graph with information. Initially the information is stored at $\source$ at time 1 and then, at every step, every informed vertex informs all its neighbors as soon as the edge between them becomes available. 
In the second problem (\Pt) we try to transfer a package with a tangible good from $\source$ to $\sink$. Now, at every step we need to decide for the actual route of the package through the network: when an edge appears, should we ship the package along it or rather wait where we currently are? 
\Pt\ is more relevant to real-life applications than \Po, where an actual \emph{good} journey needs to be found in real time.

\medskip
\noindent\textbf{Our contribution.} 
In the first part of the paper, in Section~\ref{section-Po}, we provide our results for the problem \Po, 
i.e.~for computing the expected arrival time of a foremost $\source$-$\sink$ journey in a stochastic temporal graph. 
First we prove in Section~\ref{sec:hardness_for_Po} that \Po\ is \#P-hard even for the memoryless model (and 
thus also for the memory-$k$ model, for every $k\geq 1$). The reduction is done from the problem \#PP2DNF 
which counts the number of satisfying assignments in a positive partitioned 2-DNF Boolean formula~\cite{PB83}.

Second, we provide in Section~\ref{FPTAS-subsection} a non-trivial approximation scheme for \Po, 
based on dynamic programming, for the memoryless model in the case where the 
underlying graph~$G$ is a series-parallel graph. 
More specifically, it turns out that this is a \emph{Fully Polynomial-Time Approximation Scheme (FPTAS)} 
whenever the probabilities $p_e$ are lower bounded by $\frac{1}{n^c}$ for some $c\geq 1$. 
Let $X$ be the random variable that expresses the arrival time of a foremost $\source$-$\sink$ journey. 
For every $\varepsilon \in (0,1]$, our FPTAS gives an algorithm that produces a value $\mu$ 
where $\mathbb{E}(X)-\varepsilon \leq \mu \leq \mathbb{E}(X)$, and runs in polynomial time in both $n$ and $\frac{1}{\varepsilon}$. 
Although our main result of Section~\ref{FPTAS-subsection} concerns series-parallel graphs, 
we actually present a more general FPTAS approach (see Theorem~\ref{th:generalFPTAS}) which is 
of independent interest and could lead to FPTASs also for more general classes of underlying graphs $G$.

Third, we present in Section~\ref{sec:FPRAS} a \emph{Fully Polynomial Randomized Approximation Scheme (FPRAS)} for \Po\ in the memory-$k$ model, for every $k\geq 0$, under the assumption that every edge 
appearance probability is lower bounded by $\frac{1}{n^c}$ for some $c\geq 1$. 
Let $X$ be the random variable that expresses the arrival time of a foremost $\source$-$\sink$ journey. 
For every $\varepsilon \in (0,1)$, our FPRAS gives a randomized algorithm that produces 
an estimate $\widetilde{X}$ where $(1-\varepsilon)\mathbb{E}(X)\leq \widetilde{X}\leq (1+\varepsilon)\mathbb{E}(X)$ 
with probability tending to 1 as $n\rightarrow \infty$, 
and runs in polynomial time in both $n$ and $\frac{1}{\varepsilon}$.

In the second part of the paper, in Section~\ref{sec:best-policy}, we provide our results for 
the problem \Pt, i.e.~for computing the expected arrival time of a best policy for choosing a particular 
$\source$-$\sink$ journey. 
Initially we provide in Section~\ref{sec:dynamic_prog_memoryless} a dynamic programming algorithm 
for the memoryless model which runs in $O(n^2)$ time and space. 
In wide contrast, we prove in Section~\ref{sec:hardness_for_Pt} that \Pt\ becomes \#P-hard 
for the memory-$k$ model, where $k\geq 3$, again by providing a reduction from the problem \#PP2DNF.
Finally, we provide in Section~\ref{sec:bestpolicy} a formulation of \Pt\ in the memory-$k$ model using 
the general \emph{Markov Decision Process (MDP)} framework which allows us to devise in Section~\ref{sec:doublyexp} 
an exact doubly exponential-time algorithm with running time $O(2^{(kmn+n\log n)\cdot 2^{km}})$.

\section{Preliminaries}\label{sec:prelim}

In this paper we consider temporal graphs (see Definition~\ref{temporal-def}) in which the underlying 
(static) graph $G=(V,E)$ has $n$ vertices and $m$ edges. A subgraph $H=(V,E_H)$ of $G$, 
denoted by~$H\subseteq G$, is a graph where $E_H\subseteq E$. 
For every vertex $u\in V$, the \emph{neighborhood} $\Gamma_{G}(u)$ of $u$ in $G$ is the set of adjacent vertices of $u$ in $G$. The \emph{closed neighborhood} $\Gamma_{G}[u]$ also contains vertex~$u$ itself, 
i.e.~$\Gamma_{G}[u] = \Gamma_{G}(u) \cup \{u\}$. 
For simplicity of notation we denote $[n] = \{1,2,\ldots,n\}$ for every $n\in \mathbb{N}$. 
Furthermore, sometimes we refer to the discrete time steps $t=1,2,\ldots$ as \emph{days}. 
Throughout the paper we consider stochastic temporal graphs that exhibit an edge-centric evolution, 
i.e.~every edge $e$ of $G$ is assigned one probability distribution for its appearance over time, 
independently of all other edges. 
We investigate the case where there is a ``memory effect'' that governs the probability of appearance of 
every edge over time. We distinguish now the cases where the the memory is zero or non-zero.

\begin{description}
	\item[Memoryless (or memory-0) model.] Every edge $e \in E$ evolves stochastically and independently of other edges as follows: at every time step $t\in \mathbb{N}$, $e$ appears in $G_t$ with probability~$p_e$ and is absent with probability $1-p_e$, independently of any other time step. 
The numbers $\{p_e: e \in E\}$ are given parameters of the model. 
We denote this (memoryless) stochastic temporal graph by $\mathcal{G}^{(0)} =	(G, \{ p_e  :  e \in E \} )$ or simply $\mathcal{G}^{(0)} =	(G, \{p_e\} )$.

	\item[Memory-$k$ model.] This model of temporal graphs exhibits stochastic time-dependency of the edges: 
we assume an initial (arbitrary) sequence of $k$ snapshots, $G_{-k+1}, \ldots, G_{-1}, G_0 \subseteq G$. At every time step $t\geq 1$, every edge $e$ appears independently of all other edges with 
probability that depends only on (the edge and) the history of appearance of $e$ in the $k$ previous 
snapshots. 
At every time step $t$, this history is a $k$-bit binary vector, where a $0$-entry (resp.~$1$-entry) on the $i$-th position denotes absence (resp.~appearance) of $e$ in $E_{t-k+i-1}$, for $i=1,\ldots,k$. 
Therefore the snapshot $G_t$ is the graph that appears at time $t\geq 1$ as the result of the following experiment: 
given the history $H_e^{(k)}$ of the appearance of edge $e\in E$ in the last $k$ snapshots, 
$e$ belongs to $E_t$ independently with probability~$p_e(H_e^{(k)})$. 
We denote the memory-$k$ stochastic temporal graph by $\mathcal{G}^{(k)}$.

\bigskip

	In the particular case where $k=1$, the memory-$1$ stochastic temporal graph $\mathcal{G}^{(1)}$ is the sequence $\{G_t=(V,E_t) :t \in \mathbb{N}\}$ of snapshots 
such that $E_t = \{ e \in E: X_t^e = 1 \}$, where $ \{ X_t^e\}_{t \in \mathbb{N}}$ 
is a Markov chain for the edge $e \in E$ with states $\{0,1\}$ 
(corresponding to non-appearance and appearance of $e$, respectively) and probability transition matrix:
	\begin{equation*}
	M_e = \left( 
	\begin{array}{r|cc}
	   & 0 & 1\\ \hline
	 0 &1-p_e & p_e \\
	 1 &  q_e & 1-q_e
	\end{array}
	\right) \text{, where } 0\leq p_e,q_e\leq 1.
	\end{equation*}
Using this formalism, $p_e$ (resp.~$q_e$) is the probability that the edge $e$ changes its current state 
from absence to appearance (resp.~from appearance to absence) in the next snapshot. 	
Note here that, setting $p_e=p$ and $q_e=q$ for every edge $e$, we obtain exactly the well-established 
\emph{edge-Markovian evolving graph} model introduced by Clementi et al.~\cite{clementiMMPS10}. 
\end{description}

\subsection{The problems}\label{sec:problems}

This work studies two main problems, each under the models of stochastic temporal graphs defined above. 
To describe both of these problems, let us first recall that information in temporal graphs flows via journeys, i.e.~temporal paths.

\begin{definition}[Time-edge]
	A time-edge in a temporal graph $\mathcal{G}= \{G_t: t \in \mathbb{N}\}$ is a pair $(e,t)$ such that $e \in E_t$.
\end{definition}

\begin{definition}[Journey / temporal path]\label{def:journey}
	Let $\mathcal{G} = \{G_t: t \in \mathbb{N}\}$ be a temporal graph and $\source,\sink$ be two vertices of $G$. 
	An \emph{$\source$-$\sink$ journey} (or an \emph{$\source$-$\sink$ temporal path}) in~$\mathcal{G}$ is a sequence $\big( (e_1,t_1),\ldots,(e_x,t_x) \big)$ of time-edges over a path $(e_1,\ldots,e_x)$ in $G$, where $t_1<t_2<\ldots<t_x$. 
The \emph{arrival time} of the journey is the time $t_x$ of appearance of its last edge. 
\end{definition}

\begin{definition}[Foremost Journey]\label{def:foremost}
	A foremost $\source$-$\sink$ journey in a temporal graph $\mathcal{G}$ is an $\source$-$\sink$ journey with minimum arrival time amongst all $\source$-$\sink$ journeys in $\mathcal{G}$.
\end{definition}

Notice that the arrival time of a foremost $\source$-$\sink$ journey in a stochastic temporal graph is a random variable, which we henceforth denote by $X(\source, \sink)$.
The first problem that we study here is how to compute the expected value of the latter, namely $\mathbb{E}[X(\source,\sink)]$.

\begin{problem}[\Po]\label{prob:P1}
	Given a stochastic temporal graph on an underlying graph $G=(V,E)$ and two distinct vertices $\source, \sink \in V$, compute the expected value of the arrival time of a foremost $\source$-$\sink$ journey, i.e.~$\mathbb{E}[X(\source,\sink)]$.
\end{problem}

Now suppose that an individual (say Alice) is at day 0 at vertex $\source$ and would like to arrive 
at vertex $\sink$ through a temporal path as quickly as possible. 
Denote by $s_t$ the vertex where she is located at time $t$; then $s_0=\source$. 
Every day $t$ Alice ``wakes up'' in the morning and looks at which edges are available 
in today's snapshot; 
by only knowing her current position, the history of the last $k$ snapshots, and the input parameters of the stochastic temporal graph 
(i.e.~the probabilistic rules of edge appearance), 
Alice needs to decide whether: 
\begin{enumerate}
\item[(i)] to stay at the vertex $s_t$ she currently is, or 
\item[(ii)] to use an edge of $G_{t}$ to move	to a neighboring vertex.
\end{enumerate}
That is, $s_{t+1}$ is either equal to $s_t$ or equal to some vertex of $\Gamma_{G_{t}} (s_t)$.

A natural problem we can study here is to compute the expected arrival time of an $\source$-$\sink$ journey 
that Alice can follow, using a \emph{best policy}\footnote{We use the term ``policy'' here 
(instead of ``strategy'') since, as we will see later, this problem can be 
formulated using a Markov Decision Process (MDP).} possible, i.e.~a policy (sequence of actions) that minimizes her expected arrival time at $\sink$.
Notice that the arrival time of the journey suggested to Alice by the best policy is a random variable $Y(\source, \sink)$, whose distribution depends on the specific stochastic temporal graph. 
In particular, in the memoryless model, the expectation of $Y(\source, \sink)$ depends only on the edges' probabilities of appearance. In the memory-$k$ model, the expectation of $Y(\source, \sink)$ also depends 
on the initial snapshots $G_{-k+1}, \ldots, G_{-1}, G_0$.

\begin{problem}[\Pt]\label{prob:P2}
	Given a stochastic temporal graph $\mathcal{G}^{(k)}$ on an underlying graph $G=(V,E)$ and two distinct vertices $\source, \sink \in V$, compute $\mathbb{E}_{\mathcal{G}^{(k)}}[Y(\source,\sink)]$.
\end{problem}

In particular, we will write $h(\source, \sink) \stackrel{\text{def}}{=} \mathbb{E}_{\mathcal{G}^{(0)}}[Y(\source,\sink)]$ and $h(\source, \sink, G_0) \stackrel{\text{def}}{=} \mathbb{E}_{\mathcal{G}^{(1)}}[Y(\source,\sink)]$.

\paragraph*{Difference between the two problems.}\label{sec:difference}

Before we proceed further, we first give an example illustrating that the problems \Po\ and \Pt\ are different. 
To demonstrate this, assume the memoryless model $\mathcal{G}^{(0)}$ and consider the 4-cycle $a,b,c,d,a$ as the underlying graph. Let  $\source=a$ and $\sink=c$ and assume that, at any time step, each edge appears independently with probability $\frac{1}{2}$.

Any best policy for Alice will wait until an edge incident to $a$ appears and then cross it; if both adjacent edges $(a,b)$ and $(a,d)$ appear at the same time, then it does not matter which one she chooses.
The event ``some edge adjacent to $a$ appears'' occurs with probability $\frac{3}{4}$, hence, the expected time until such an edge appears is $\frac{4}{3}$. Furthermore, when Alice reaches one of the vertices $b$ or $d$, an optimal policy will never suggest going back to $a$, so Alice will have to wait until the last edge to $c$ appears, which takes $2$ steps on expectation. Overall, the optimal policy for Alice will take $h(a,c)= \frac{10}{3}$ steps on expectation. This is the solution to \Pt\ (see Problem~\ref{prob:P2}). 

On the other hand, \Po\ (see Problem~\ref{prob:P1}) asks for the expectation of the arrival time $X(a,c)$ 
of a foremost $\source$-$\sink$ journey. 
To compute $\mathbb{E}[X(a,c)]$, denote by $T_b$ (resp.~$T_d$) the arrival time of a journey allowed to use only edges $(a,b)$ and $(b,c)$ (resp.~$(a,d)$ and $(d,c)$), when they appear. Then, 
\[
X(a,c) > k \ \ \Leftrightarrow \ \ T_b>k \ \text{ and } \ T_d>k
\]
But the probability of the event $\{T_b>k\}$ is equal to the probability that either $(a,b)$ does not appear until (and including) step $k$ plus the probability that it appears within the first $k$ steps, and $(b,c)$ does not appear after that until (and including) $k$. Therefore, 
\[\Pr[T_b>k] = \frac{1}{2^k}+k\frac{1}{2} \frac{1}{2^{k-1}}=(k+1)\frac{1}{2^k}.\]
By symmetry we have $\Pr[T_b>k]=\Pr[T_d>k]$ and, by independence, for any $k \geq 2$:  
\[
\Pr[X(a,c) > k] = \Pr[T_b>k] \Pr[T_d>k] = \frac{1}{2^{2k}} +k \frac{1}{2^{2k-1}}+k^2 \frac{1}{2^{2k}}.
\]    
By using the fact that $\mathbb{E}[X(a,c)] = \sum_{k=0}^{\infty} \Pr[X(a,c) > k] = 2+ \sum_{k=2}^{\infty} \Pr[X(a,c) > k]$, it follows that $\mathbb{E}[X(a,c)] = 2+\frac{26}{27}= \frac{80}{27}$, which is strictly smaller than $\frac{10}{3}$.

In fact, the gap between the solution to \Po\ and the solution to \Pt\ can be arbitrarily large: Consider the graph consisting of vertices $\source$ and $\sink$ and $n-2$ vertex disjoint paths of length 2 between $\source$ and $\sink$. Assume also that, under the~memoryless model, every edge incident to $\source$ appears each day with probability $1$ and every edge incident to $\sink$ appears each day independently with probability $n^{-0.9}$. 
Similarly to the above example, the expected arrival time of a best policy for Alice 
is $h(s, y) = 1+n^{0.9}$. 
On the other hand, the arrival time of the foremost journey from $\source$ to $\sink$ will be equal to the first day 
after day 1 on which some edge incident to $\sink$ appears. 
But the time needed for the latter to happen follows the geometric distribution with success probability $1-(1-n^{-0.9})^{n-2} = 1-o(1)$. Therefore, the expected arrival time of the foremost journey will be $\mathbb{E}[X(s,y)] = 2+o(1)$, i.e.~much smaller than $h(s, y) = 1+n^{0.9}$.

As a final note, the expected arrival time $\mathbb{E}[X(s,y)]$ of the foremost $\source$-$\sink$ journey 
is always upper-bounded by the minimum among the expected values of the arrival times of all $\source$-$\sink$ journeys in the temporal graph. This is actually implied by a more general and well-known lemma 
in Probability Theory (Fatou's lemma~\cite[p.~29]{durrett_book}) 
which establishes that the expected value of the minimum among $n$ random variables is upper-bounded by the minimum among all the variables' expectations.

\section{Computing the expected minimum arrival time}\label{section-Po}

\subsection{Hardness of exact computation in the memoryless model}\label{sec:hardness_for_Po}

In this section we show that, even in the memoryless model, \Po\ is {\#P}-hard in both undirected graphs 
and directed acyclic graphs (DAGs).
In the proof of the following theorem, the edges can be treated either as oriented, in which case we obtain the result for DAGs,
or as non-oriented, in which case we obtain the result for undirected graphs.

\begin{theorem}\label{thm:prob_1_hard}
	\Po\ in the memoryless model is {\#P}-hard.
\end{theorem}
\begin{proof}
	To prove the theorem we will provide a reduction from the {\#P}-complete problem \#PP2DNF \cite{PB83}. 
	The latter problem is defined as follows. Let $X = \{x_1, x_2, \ldots, x_n\}$ and $Y = \{y_1, y_2, \ldots, y_m\}$
	be two disjoint sets of Boolean variables. A \textit{positive}, \textit{partitioned} 2-DNF formula is a DNF formula of the form:
	\[
	\Phi = \bigvee\limits_{(i,j) \in E} x_i y_j,
	\] 
	for some $E \subseteq [n] \times [m]$.
	Given a positive, partitioned 2-DNF formula $\Phi$, the problem \#PP2DNF asks for the number 
	of truth assignments satisfying $\Phi$.
	Let $\Phi $ be an instance of \#PP2DNF. We define $G$ to be a graph
	with the vertex set $\{\source, \sink\} \cup X \cup Y$ and the edge set 
	$\{(\source,x_i) ~|~ x_i \in X \} \cup \{ (x_i,y_j) ~|~ (i,j) \in E \} \cup \{ (y_i,\sink) ~|~ y_j \in Y \}$, see Figure~\ref{fig:pp2dnf}.
	
\begin{figure}[ht]
\centering
\includegraphics[width=0.45\textwidth]{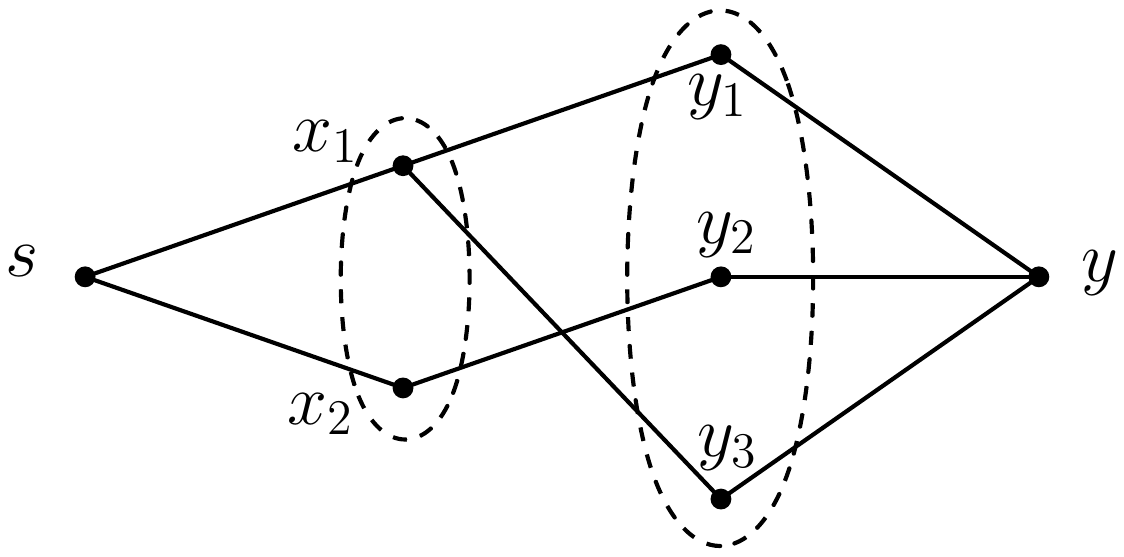}
\caption{Example construction of $G$, given the positive, partitioned 2-DNF formula $\Phi=(x_1y_1) \vee (x_1y_3) \vee (x_2y_2)$.}
\label{fig:pp2dnf}
\end{figure}
	
	First we claim\footnote{This claim was provided by Antoine Amarilli (\url{https://cstheory.stackexchange.com/q/42239}).}
	that the number $\psi$ of satisfying assignments of $\Phi$ is equal to the number of spanning subgraphs
	of $G$ which contain all the edges from $\{ (x_i,y_j) ~|~ (i,j) \in E \}$ and have a simple path from $\source$ to $\sink$ of length 3.
	To see the claim, for every subset $S \subseteq \{ (\source,x_i) ~|~ x_i \in X \} \cup \{ (y_i,\sink) ~|~ y_j \in Y \}$ of edges we define 
	a truth assignment $\alpha$ that assigns $x_i = 1$ iff $(\source,x_i) \in S$ and $y_j = 1$ iff $(y_j,\sink) \in S$. 
	Notice that every $\source$-$\sink$ path of length 3 in $G$ is of the form $(\source,x_i,y_j,\sink)$ for some $(i,j) \in E$.
	Therefore, if the subgraph spanned by $S$ contains a path $(\source,x_i,y_j,\sink)$, then $\alpha$ assigns 1 to both $x_i$ and $y_j$,
	and hence $\alpha$ satisfies $\Phi$.
	Conversely, given an assignment $\alpha$ satisfying $\Phi$, we define a subgraph of $G$ spanned by the edge set 
	$\{ (\source,x_i) ~|~ x_i \text{ is assigned 1 by } \alpha \} \cup \{ (x_i,y_j) ~|~ (i,j) \in E \} \cup \{ (y_i,\sink) ~|~ y_j \text{ is assigned 1 by } \alpha \}$. Since $\alpha$ is satisfying assignment, there exists $(i,j) \in E$ such $\alpha$ assigns 1 to both $x_i$ and $y_j$, and therefore
	the subgraph contains the $\source$-$\sink$ path $(\source,x_i,y_j,\sink)$ of length 3.
	
	Now we define an instance of \Po\ in the memoryless model as follows. Let $H$ be the graph obtained from $G$
	by adding three new vertices $v_1,v_2,v_3$ and four new edges $(\source,v_1), (v_1,v_2), (v_2,v_3), (v_3,\sink)$, which all together form
	a new $\source$-$\sink$ path of length 4. 
	For every edge $e \in \{ (\source,x_i) ~|~ x_i \in X \} \cup \{ (y_i,\sink) ~|~ y_j \in Y \}$
	we set $p_e = 1/2$, and for any other edge $e$ of $H$ we set $p_e=1$. 
	In this stochastic temporal graph the duration of a foremost journey from $\source$ to $\sink$ is either
	3, if for some $(i,j) \in E$ the edge $(\source,x_i)$ appears in time slot 1, and the edge $(y_j,\sink)$ appears in time slot 3, or
	4 otherwise.
	In other words, the duration of a foremost $\source$-$\sink$ journey depends only on the subgraph of $G$ spanned by the edge set
	$R_1 \subseteq \{ (\source,x_i) ~|~ x_i \in X \}$ that appears in slot 1, and by the edge set $R_3 \subseteq \{ (y_i,\sink) ~|~ y_j \in Y \}$ 
	that appears in slot 3. The duration is equal to 3 if and only if the subgraph of $G$ spanned by 
	$R_1 \cup \{ (x_i,y_j) ~|~ (i,j) \in E \} \cup R_3$ has an $\source$-$\sink$ path of length 3. Since every edge in $R_1 \cup R_3$ appears independently
	with probability $1/2$, it follows that the probability that this subgraph 
	has a path of length 3 is equal to $p=\frac{\psi}{2^{n+m}}$. Consequently,
	\begin{equation}
	\notag
	\mathbb{E}[X(\source,\sink)] = 3p + 4(1-p)  = 4-p,
	\end{equation}	
	and hence $\psi = 2^{n+m} (4-\mathbb{E}[X(\source,\sink)])$. 
	Therefore, knowing the expected duration $\mathbb{E}[X(\source,\sink)]$ of an $\source$-$\sink$ foremost journey, we can efficiently compute the number of satisfying assignments of $\Phi$,
	which proves that the computation of $\mathbb{E}[X(\source,\sink)]$ is {\#P}-hard.
\end{proof}

\begin{corollary}\label{cor:prob_1_hard_memory}
For every $k\geq 0$, \Po\ in the memory-$k$ model is {\#P}-hard.
\end{corollary}

\subsection{The FPTAS for the memoryless model on series-parallel graphs}
\label{FPTAS-subsection}

\subsubsection{The case of paths}\label{sec:paths_Po}

In this section we will consider a stochastic temporal graph $\mathcal{P}^{(0)} = (P=(V,E), \{ p_e\} )$
with the underlying graph being a path $P=(\source = v_0, v_2, \ldots, v_{n} = \sink)$.

\begin{lemma}\label{lem:pathExpected}
	$\mathbb{E}[X_{\mathcal{P}^{(0)}}(\source,\sink)] = \sum_{e \in E} \frac{1}{p_e}$.
\end{lemma}
\begin{proof}
	Consider a stochastic temporal graph with a single edge $e$ which appears every day independently with probability $p_e$, and
	let $X_e$ be a random variable equal to the duration of the foremost journey 
	from one of the endpoints of $e$ to the other. 
	Then $X_{\mathcal{P}^{(0)}}(\source,\sink) = \sum_{e \in E} X_e$.
	Notice that $X_e$ is a geometric random variable with probability mass function 
	$\Pr[X_e = i] = (1-p_e)^{i-1}p_e$ for $i=1,2,3,...$, and expectation $\mathbb{E}[X_e] = \frac{1}{p_e}$.
	Therefore $\mathbb{E}[X_{\mathcal{P}^{(0)}}(\source,\sink)] = \sum_{e \in E} \mathbb{E}[X_e] = \sum_{e \in E} \frac{1}{p_e}$.
\end{proof} 

\noindent
Let us denote by $\mu$ the expectation $\mu \stackrel{\text{def}}{=} \mathbb{E}[X_{\mathcal{P}^{(0)}}(\source,\sink)] = \sum_{e \in E} \frac{1}{p_e}$.
Note that
\begin{equation}\label{eq:muSum}
\mu = \sum_{i=1}^{\infty} \Pr[X_{\mathcal{P}^{(0)}}(\source,\sink) \geq i]. 
\end{equation}
In the remainder of this section we will show that the first $O(\mu \ln \mu)$ terms of sum (\ref{eq:muSum}) already 
give a very good approximation of $\mu$. In our analysis we will use the following bound.

\begin{theorem}[\cite{janson2018tail}]\label{th:tailSumGeom}
	Let $X = \sum_{i=1}^n X_i$, where $n \geq 1$ and $X_i$, $i = 1, \ldots, n,$ are independent
	geometric random variables with parameters $p_1, p_2, \ldots, p_n \in (0,1]$, respectively.
	Let $\mu = \mathbb{E}[X] = \sum_{i=1}^n \frac{1}{p_i}$. Then for any $\lambda \geq 1$,
	\[
	\Pr[X \geq \lambda \mu] \leq e^{1-\lambda}.
	\]
\end{theorem}

\begin{lemma}\label{lem:pathExpectedApprox}
	Let $\varepsilon$ be a number such that $0 < \varepsilon \leq 1$. Then
	\begin{equation}\label{eq:pathExpApprox}
	\mu - \sum_{i=1}^{\tau} \Pr[X_{\mathcal{P}^{(0)}}(\source,\sink) \geq i] = 
	\sum_{i=\tau+1}^{\infty} \Pr[X_{\mathcal{P}^{(0)}}(\source,\sink) \geq i]< \varepsilon,
	\end{equation}
	for every $\tau \geq \mu \left( \ln \frac{\mu}{\varepsilon} + 1 \right)$, where $\mu = \mathbb{E}[X_{\mathcal{P}^{(0)}}(\source,\sink)]$.
\end{lemma}
\begin{proof}
	The equality in (\ref{eq:pathExpApprox}) follows from (\ref{eq:muSum}). In the rest of the proof we show the inequality.
	Since $\tau \geq \mu \left( \ln \frac{\mu}{\varepsilon} + 1 \right) \geq \mu$, using Theorem~\ref{th:tailSumGeom} we have

	\begin{equation*}
	\begin{split}
	\sum_{i=\tau+1}^{\infty} \Pr[X_{\mathcal{P}^{(0)}}(\source,\sink) \geq i] & \leq  
	\sum_{i=\tau+1}^{\infty} e^{1-\frac{i}{\mu}} 
	= \frac{e^{1-\tau/\mu}}{e^{1/\mu} - 1} \leq
	\frac{e^{1-\mu \left( \ln \frac{\mu}{\varepsilon} + 1 \right)/\mu}}{e^{1/\mu}-1} =\\
	&=\frac{\varepsilon}{\mu (e^{1/\mu}-1)} \leq
	\frac{\varepsilon}{\mu (1 + \frac{1}{\mu} + \frac{1}{2\mu^2} - 1)} = 
	\frac{\varepsilon}{1 + \frac{1}{2\mu}} < \varepsilon,
	\end{split}
	\end{equation*}
	where we used the inequality $e^x \ge 1 + x + x^2/2$ which holds for every $x \geq 0$.
\end{proof}

\subsubsection{A general FPTAS approach}
\label{sec:generalFPTAS}

While deriving analytically and computing efficiently the exact solution of \Po\ in a path
is an easy task (cf. Lemma~\ref{lem:pathExpected}), it does not seem to be trivial for a slight generalization of paths,
called \textit{parallel compositions of paths}. A \text{parallel composition of paths} is the graph obtained 
from a collection of disjoint paths $P_1, P_2, \ldots, P_\ell$ with end vertices $\source_i, \sink_i$, $i =1, \ldots, \ell$,
respectively, by identifying the vertices $\source_1, \source_2, \ldots, \source_\ell$ in a single vertex $\source$, and
by identifying the vertices $\sink_1, \sink_2, \ldots, \sink_\ell$ in a single vertex $\sink$.

It is not clear whether there exists an efficient procedure for computing the expected arrival time from $\source$ to $\sink$
in a parallel composition of paths, even if the parallel paths are of equal length and all the probabilities of edge appearance are the same.
In this section we present a general approach for developing \emph{$\varepsilon$-additive approximation algorithms}\footnote{A feasible solution is \emph{$\varepsilon$-additive approximate} if it is within $\varepsilon$ additive factor from the optimal value.
	An algorithm is called an \emph{$\varepsilon$-additive approximation algorithm} if it returns an $\varepsilon$-additive approximate
	solution for any instance.}
for computing the expected arrival time of a foremost journey in special classes of stochastic temporal graphs.
In Section~\ref{sec:FPTASforSP} we apply this approach to develop an efficient $\varepsilon$-additive approximation algorithm
for the problem on the class of stochastic temporal graphs with underlying graphs being series-parallel graphs, which generalize
parallel compositions of paths and graphs, in which all simple $\source$-$\sink$ paths are of the same length.

Throughout the section we denote by $\mathcal{G}^{(0)} = (G=(V,E), \{ p_e \} )$ 
a memoryless stochastic temporal graph with $n$ vertices and $m$ edges, and by $\source, \sink \in V$ two distinct vertices in $G$. 
Furthermore, we denote by $H=(V,E,w)$ the weighted graph obtained from the underlying graph $G$ by assigning
to every edge $e \in E$ the weight $w(e) = \frac{1}{p_e}$.

\begin{definition}
\label{def:temporal-subgraph}
Let $\mathcal{G}^{(0)}$ be a memoryless stochastic temporal graph, where $G$ is the underlying graph. 
A \emph{stochastic temporal subgraph} $\mathcal{H}^{(0)}$ of $\mathcal{G}^{(0)}$ is a stochastic temporal 
graph which has a subgraph $H\subseteq G$ as an underlying graph and inherits 
all edge appearance probabilities from $\mathcal{G}^{(0)}$.
\end{definition}

\begin{observation}\label{obs:subgraph}
Let $\mathcal{H}^{(0)}$ be a \emph{stochastic temporal subgraph} of the stochastic temporal graph $\mathcal{G}^{(0)}$. 
Then for every natural number $i$ we have 
$\Pr[X_{\mathcal{G}^{(0)}}(\source,\sink) \geq i] \leq \Pr[X_{\mathcal{H}^{(0)}}(\source,\sink) \geq i]$,
and hence $\mathbb{E}[X_{\mathcal{G}^{(0)}}(\source,\sink)] \leq \mathbb{E}[X_{\mathcal{H}^{(0)}}(\source,\sink)]$.
\end{observation}

\noindent
The following lemma is a direct consequence of Observation~\ref{obs:subgraph} and Lemma~\ref{lem:pathExpected}.

\begin{lemma}
	Let $w^*$ be the minimum weight of an $\source$-$\sink$ path in $H$. 
	Then $\mathbb{E}[X_{\mathcal{G}^{(0)}}(\source,\sink)]  \leq w^*$.
\end{lemma}

\begin{theorem}\label{th:generalFPTAS}
	Let $c\in \mathbb{N}$ and $\varepsilon \in (0,1]$. 
	Let $p_e \geq \frac{1}{n^c}$ for every $e \in E$ and suppose that there exists an algorithm $A$ 
	that computes in time $O\left( f(\ell,n,m) \right)$ the probabilities $\Pr[X_{\mathcal{G}^{(0)}}(\source,\sink) \geq i]$ for all $i=1,\ldots,\ell$.
	Then there exists an algorithm $B$ that approximates $\mathbb{E}[X_{\mathcal{G}^{(0)}}(\source,\sink)]$
	within the additive factor of $\varepsilon$ in time 
	\[
	O\Big( f\left( n^{c+1} \ln \frac{n}{\varepsilon}, n, m \right) + n \ln n + m\Big).
	\]
	Consequently, if $f(\ell,n,m)$ is a polynomial in variables $\ell,n$, and $m$, then $B$ is an FPTAS on the instance 
	$(\mathcal{G}^{(0)}, \source, \sink)$.
\end{theorem}
\begin{proof}
	Let $P = (s=v_0, v_1, \ldots, v_r = y)$ be a minimum weight $\source$-$\sink$ path in $H$, and let $\mathcal{P}^{(0)}$ be the 
	stochastic temporal subgraph of $\mathcal{G}^{(0)}$ restricted to the vertices of $P$. 
	For convenience, let us denote $e_i = v_{i-1}v_i$ for every $i=1, \ldots, r$.
	Then, by definition and Lemma~\ref{lem:pathExpected}, the weight $w^*$ of $P$ is equal to 
	$\sum_{i=1}^{r} \frac{1}{p_{e_i}} = \mathbb{E}[X_{\mathcal{P}^{(0)}}(\source,\sink)]$.
	Let $\tau := w^* \left( \ln \frac{w^*}{\varepsilon} + 1 \right)$.
	Then, by Observation~\ref{obs:subgraph} and Lemma~\ref{lem:pathExpectedApprox}, we have that
	
	\[
	\sum_{i=\tau+1}^{\infty} \Pr[X_{\mathcal{G}^{(0)}}(\source,\sink) \geq i] \leq 
	\sum_{i=\tau+1}^{\infty} \Pr[X_{\mathcal{P}^{(0)}}(\source,\sink) \geq i] < \varepsilon,
	\]
	and hence
\begin{eqnarray*}
\sum_{i=1}^{\tau }\Pr [X_{\mathcal{G}^{(0)}}(\source,\sink)\geq i] \
&\leq& 
\mathbb{E}[X_{\mathcal{G}^{(0)}}(\source,\sink)] 
 \ = \ \sum_{i=1}^{\infty}\Pr [X_{\mathcal{G}^{(0)}}(\source,\sink)\geq i] \\
&<&\sum_{i=1}^{\tau }\Pr [X_{\mathcal{G}^{(0)}}(\source,\sink)\geq i]+\varepsilon, 
\end{eqnarray*}
that is, $\sum_{i=1}^{\tau} \Pr[X_{\mathcal{G}^{(0)}}(\source,\sink) \geq i]$
approximates $\mathbb{E}[X_{\mathcal{G}^{(0)}}(\source,\sink)]$ within the additive factor of $\varepsilon$.
	
	\medskip
	\noindent
	Now we define the desired algorithm $B$ as follows:
	\begin{enumerate}
		\item Construct the graph $H$ and compute the minimum weight $w^*$ of an $\source$-$\sink$ path in $H$
		using Dijkstra's algorithm.
		\item Using algorithm $A$, compute the probabilities $\Pr[X_{\mathcal{G}^{(0)}}(\source,\sink) \geq i]$ for every $i=1,\ldots,\tau$,
		where $\tau = w^* \left( \ln \frac{w^*}{\varepsilon} + 1 \right)$.
		\item Output $\sum_{i=1}^{\tau} \Pr[X_{\mathcal{G}^{(0)}}(\source,\sink) \geq i]$.
	\end{enumerate}
	
	The above discussion implies that algorithm $B$ correctly computes the declared approximation
	of $\mathbb{E}[X_{\mathcal{G}^{(0)}}(\source,\sink)]$. It remains to justify the time complexity.
	First, Dijkstra's algorithm can be implemented to work in time $O(n \ln n + m)$ \cite{fredman1987fibonacci}.
	Second, the assumption on $p_e$'s implies that $w^* = O(n^{c+1})$, and hence
	$\tau = w^* \left( \ln \frac{w^*}{\varepsilon} + 1 \right) = O\left( n^{c+1} \ln \frac{n}{\varepsilon} \right)$.
	Therefore the assumption of the theorem implies that the last two steps of the algorithm
	run in time $O\Big( f\big( n^{c+1} \ln \frac{n}{\varepsilon}, n, m \big) \Big)$, which in turn implies the complexity bound
	and completes the proof.
\end{proof}

\subsubsection{The FPTAS for stochastic temporal series-parallel graphs}
\label{sec:FPTASforSP}

In the present section we use the approach from Section~\ref{sec:generalFPTAS} to derive 
a polynomial-time approximation scheme for stochastic temporal series-parallel graphs.
According to Theorem~\ref{th:generalFPTAS}, the development of such an algorithm reduces
to the design of an efficient procedure of computing probabilities of the form $\Pr[X_{\mathcal{G}^{(0)}}(\source,\sink) \geq i]$,
which is the main goal of this section.

Let $G$ be a graph and $\source$ and $\sink$ be two distinct vertices in $G$.
The triple $(G,\source,\sink)$ is a \textit{two-terminal series-parallel} graph, with terminals $\source$ and $\sink$, 
if $G$ can be constructed by a sequence of the following two operations starting from a set of copies of a single-edge
two-terminal series-parallel graph $(K_2, a, b)$.
\begin{enumerate}
	\item \textit{Parallel composition}: Given a pair of two-terminal series-parallel graphs $(H_1, \source_1, \sink_1)$ and 
	$(H_2, \source_2, \sink_2)$, form a new two-terminal series-parallel graph $(G,\source,\sink)$ by
	identifying $\source = \source_1 = \source_2$ and $\sink = \sink_1 = \sink_2$.
	
	\item \textit{Series composition}: Given a pair of two-terminal series-parallel graphs $(H_1, \source_1, \sink_1)$ and 
	$(H_2, \source_2, \sink_2)$, form a new two-terminal series-parallel graph $(G,\source,\sink)$ by identifying
	$\source = \source_1$, $\sink_1 = \source_2$, and $\sink = \sink_2$.
\end{enumerate}

\noindent
Finally, a graph $G$ is called \textit{series-parallel} if $(G,\source,\sink)$ is a two-terminal series-parallel graph
for some pair of distinct vertices $\source$ and $\sink$ of $G$.

The sequence of parallel and series compositions leading to a two-terminal series-parallel graph $(G=(V,E),\source,\sink)$
can be conveniently represented by a decomposition tree.
A binary tree $T = (V_T, E_T)$ with a labeling function $\sigma : V_T \rightarrow \{ \texttt{s,p}\} \cup E \times \{0, 1 \}$ is called 
a \textit{decomposition tree} of a two-terminal series-parallel graph $(G,\source,\sink)$ if and only if the leaves of $T$ are
labeled with elements of $E \times \{0, 1 \}$ such that every $e \in E$ appears in exactly one label, 
internal nodes are labeled with $\texttt{p}$ or $\texttt{s}$, and $G$ can be generated recursively 
using $T$ as follows: If $T$ is a single node $v$ with $\sigma(v) = (e,\alpha)$, then $G$ consists of the single edge $e$ 
with the source being the vertex with the smallest ID, if $\alpha = 0$, and with the source being the vertex with the largest
ID, if $\alpha = 1$.
Otherwise, let $T_1$ (resp. $T_2$) be the right (resp. left) subtree of $T$ and $(H_1,\source_1,\sink_1)$ and
$(H_2,\source_2,\sink_2)$ be two-terminal series-parallel graphs with decomposition trees $T_1$ and $T_2$: 
if $\sigma(v) = \texttt{p}$ (resp. $\texttt{s}$) then $G$ is the parallel (resp. series) composition of 
$(H_1,\source_1,\sink_1)$ and $(H_2,\source_2,\sink_2)$.

We will make use of tree decompositions of series-parallel graphs in our algorithm. 
It is known that a tree decomposition of a series-parallel graph can be constructed in linear time.

\begin{theorem}[\cite{valdes1982recognition}]\label{th:SPdecomp}
	Given a two-terminal series-parallel graph with $n$ vertices and $m$ edges, its tree decomposition can be
	computed in time $O(n+m)$.
\end{theorem}

Let $\mathcal{G}^{(0)} = (G=(V,E), \{ p_e  \} )$ be a stochastic temporal graph with the underlying graph
$G$ being series-parallel. Let also $\source, \sink \in V$ be two distinct vertices in $G$ such that $(G, \source, \sink)$
is a two-terminal series-parallel graph. We will present a dynamic programming algorithm which, for a given natural number $\ell$,
computes the set of probabilities:
\begin{equation}\label{eq:setProb1}
\Pr[X_{\mathcal{G}^{(0)}}(\source,\sink) \geq i], \ \ i=1,\ldots,\ell.
\end{equation}

\noindent
For convenience, the algorithm will also support the set of probabilities:
\begin{equation}\label{eq:setProb2}
\Pr[X_{\mathcal{G}^{(0)}}(\source,\sink) = i], \ \ i=1,\ldots,\ell-1.
\end{equation}

\noindent
Notice that having computed one of the sets of probabilities, the other set can be computed in $O(\ell^2)$ time.

The algorithm is based on the following recursive equations. Since $(G, \source, \sink)$ is a two-terminal series-parallel graph,
it is either a single-edge graph, or can be obtained from smaller two-terminal series-parallel graphs 
$(H_1, \source_1, \sink_1)$, $(H_2, \source_2, \sink_2)$ via one of the two composition operations.

\begin{enumerate}
	\item In the case of a single-edge graph we have for every $i \in [\ell-1]$ that:
	\begin{equation}\label{eq:single}
	\Pr[X_{\mathcal{G}^{(0)}}(\source,\sink) = i] = (1-p)^{i-1}p,
	\end{equation}
	where $p$ is the probability of appearance of the unique edge of the graph.
	
	\item In the case of parallel composition we have for every $i \in [\ell]$ that:
	\begin{equation}\label{eq:parallel}
	\Pr[X_{\mathcal{G}^{(0)}}(\source,\sink) \geq i] = \Pr[X_{\mathcal{H}^{(0)}_1}(\source_1,\sink_1) \geq i] \cdot 
	\Pr[X_{\mathcal{H}^{(0)}_2}(\source_2,\sink_2) \geq i],
	\end{equation}
	where $\mathcal{H}^{(0)}_1$ and $\mathcal{H}^{(0)}_2$ are the stochastic temporal subgraphs of $\mathcal{G}^{(0)}$
	restricted to the vertices of $H_1$ and $H_2$, respectively.
	
	\item In the case of series composition, we have for every $i \in [\ell-1]$ that:
	\begin{equation}\label{eq:series}
	\Pr[X_{\mathcal{G}^{(0)}}(\source,\sink) = i] = \sum_{j = 1}^{i-1} 
	\Pr[X_{\mathcal{H}^{(0)}_1}(\source_1,\sink_1) = j] \cdot 
	\Pr[X_{\mathcal{H}^{(0)}_2}(\source_2,\sink_2) = i-j].
	\end{equation}	
\end{enumerate}

\begin{algorithm}[ht]
	\caption{\textsc{Compute SP probabilities}}  
	\label{alg:dynam}  
	\begin{algorithmic}[1]
		\REQUIRE{A stochastic temporal graph $\mathcal{G}^{(0)} = (G, \{ p_e   \} )$ 
		    on an underlying two-terminal series-parallel graph $(G,\source,\sink)$ with  
			a tree decomposition $T_G$, and a natural number~$\ell$.} 
		\ENSURE{The sets $\{ \Pr[X_{\mathcal{G}^{(0)}}(\source,\sink) \geq i] :  i \in [\ell] \}$ and
			$\{ \Pr[X_{\mathcal{G}^{(0)}}(\source,\sink) = i] :  i \in [\ell-1] \}$} 
		
		\medskip
		
		\IF{$G$ is a single-edge graph with the unique edge $e$}
		\FOR{$i = 1$ to $\ell-1$} \label{line:dynamForSingleStart}
		    \STATE{$\Pr[X_{\mathcal{G}^{(0)}}(\source,\sink) = i] = (1-p_e)^{i-1}p_e$}
		\ENDFOR
		\vspace{0,1cm}
		\STATE{Compute the set of probabilities $\{ \Pr[X_{\mathcal{G}^{(0)}}(\source,\sink) \geq i] :  i \in [\ell] \}$}
		\label{line:dynamForSingleEnd}
		\ELSE
		\STATE{$(G,\source,\sink)$  is a composition of two series-parallel subgraphs $(H_1,\source_1,\sink_1)$ and $(H_2,\source_2,\sink_2)$}
		\vspace{0,1cm}
		\STATE{\textsc{Compute SP probabilities}($\mathcal{H}^{(0)}_1,\source_1,\sink_1,T_{H_1},\ell$)} \label{line:dynamSPprobH1}
		\vspace{0,1cm}
		\STATE{\textsc{Compute SP probabilities}($\mathcal{H}^{(0)}_2,\source_1,\sink_1,T_{H_2},\ell$)} \label{line:dynamSPprobH2}
		\vspace{0,1cm}
		\IF{$(G,\source,\sink)$ is the parallel composition of $(H_1,\source_1,\sink_1)$ and $(H_2,\source_2,\sink_2)$}
		\FOR{$i = 1$ to $\ell$}\label{line:dynamForParallelStart}
		    \STATE{$\Pr[X_{\mathcal{G}^{(0)}}(\source,\sink) \geq i] = \Pr[X_{\mathcal{H}^{(0)}_1}(\source_1,\sink_1) \geq i] \cdot 
		\Pr[X_{\mathcal{H}^{(0)}_2}(\source_2,\sink_2) \geq i]$}
		\ENDFOR 
		\vspace{0,1cm}
		\STATE{Compute the set of probabilities $\{ \Pr[X_{\mathcal{G}^{(0)}}(\source,\sink) = i] :  i \in [\ell-1] \}$} \label{line:dynamForParallelEnd}
		\ENDIF
		\vspace{0,1cm}
		\IF{$(G,\source,\sink)$ is the series composition of $(H_1,\source_1,\sink_1)$ and $(H_2,\source_2,\sink_2)$}
		\FOR{$i = 1$ to $\ell-1$}\label{line:dynamForSeriesStart}
		    \STATE{$\Pr[X_{\mathcal{G}^{(0)}}(\source,\sink) = i] = \sum_{j = 1}^{i-1} 
		\Pr[X_{\mathcal{H}^{(0)}_1}(\source_1,\sink_1) = j] \cdot 
		\Pr[X_{\mathcal{H}^{(0)}_2}(\source_2,\sink_2) = i-j]$}
		\ENDFOR
		\vspace{0,1cm}
		\STATE{Compute the set of probabilities $\{ \Pr[X_{\mathcal{G}^{(0)}}(\source,\sink) \geq i] :  i \in [\ell] \}$} \label{line:dynamForSeriesEnd}
		\ENDIF
		\ENDIF
		
		\medskip
		
		\RETURN{$\{ \Pr[X_{\mathcal{G}^{(0)}}(\source,\sink) \geq i] :  i \in [\ell] \}$ and 
			$\{ \Pr[X_{\mathcal{G}^{(0)}}(\source,\sink) = i]  :  i \in [\ell-1] \}$}
	\end{algorithmic}
\end{algorithm}

\begin{theorem}\label{th:spProbAlg}
	Given a stochastic temporal series-parallel graph $\mathcal{G}^{(0)} = (G=(V,E), \{ p_e  \} )$, two vertices $\source, \sink \in V$
	such that $(G,\source,\sink)$ is a two-terminal series-parallel graph, and a natural number $\ell$, Algorithm~\ref{alg:dynam}  
	correctly computes the probabilities $\{ \Pr[X_{\mathcal{G}^{(0)}}(\source,\sink) \geq i]  :  i \in [\ell] \}$ and
	$\{ \Pr[X_{\mathcal{G}^{(0)}}(\source,\sink) = i]  :  i \in [\ell-1] \}$
	in time $O(m\ell^2)$, where $m = |E|$.  
\end{theorem}
\begin{proof}
	We start with the analysis of the correctness of the algorithm.
	First, if the underlying graph of the input stochastic temporal graph is  a single-edge graph, then the algorithm computes 
	the required sets of probabilities in lines~\ref{line:dynamForSingleStart}-\ref{line:dynamForSingleEnd} using
	equations~(\ref{eq:single}). Second, if the underlying graph is not a single-edge graph, then, by definition,
	$(G,\source,\sink)$ is either the parallel or the series composition of two two-terminal series-parallel graphs
	$(H_1,\source_1,\sink_1)$ and $(H_2,\source_2,\sink_2)$ whose decomposition trees are the subtrees 
	$T_{H_1}$ and $T_{H_2}$ of $T_G$ rooted at the children of the root of $T_G$.
	In the case of parallel composition, the algorithm computes the sets of probabilities in lines 
	\ref{line:dynamForParallelStart}-\ref{line:dynamForParallelEnd} using
	equations~(\ref{eq:parallel}). In the case of series composition, the algorithm computes the sets of probabilities 
	in lines~\ref{line:dynamForSeriesStart}-\ref{line:dynamForSeriesEnd} using equations~(\ref{eq:series}).
	In both cases, the computation of the probabilities uses only the corresponding sets of probabilities for
	the stochastic temporal subgraphs $\mathcal{H}^{(0)}_1$ and $\mathcal{H}^{(0)}_2$, which are computed recursively in lines 
	\ref{line:dynamSPprobH1} and~\ref{line:dynamSPprobH2}, respectively.
	
	In order to analyze the complexity of Algorithm~\ref{alg:dynam}, we observe that for every node of the decomposition tree
	of the underlying graph the algorithm makes exactly one recursive call. In each of the calls, the algorithm executes
	either lines~\ref{line:dynamForSingleStart}-\ref{line:dynamForSingleEnd}, 
	or lines~\ref{line:dynamForParallelStart}-\ref{line:dynamForParallelEnd}, 
	or lines~\ref{line:dynamForSeriesStart}-\ref{line:dynamForSeriesEnd}. It is easy to check that each of these sets of lines
	performs $O(\ell^2)$ operations of addition or multiplication. Since $T_G$ is a binary tree and has exactly $m$ leaves,
	in total $T_G$ has $2m-1$ nodes, and therefore the total time complexity of Algorithm~\ref{alg:dynam} is $O(m\ell^2)$.
\end{proof}

Finally we present an FPTAS for the expected arrival time of a foremost $\source$-$\sink$ journey in a stochastic temporal series-parallel graph.

\begin{algorithm}[ht]
	\caption{FPTAS for \Po\ in stochastic temporal series-parallel graphs} 
	\label{alg:FPTAS_SP}  
	\begin{algorithmic}[1]
		\REQUIRE{A stochastic temporal series-parallel graph $\mathcal{G}^{(0)} = (G=(V,E), \{ p_e  \} )$ 
			such that $p_e \geq \frac{1}{n^c}$ for every $e \in E$, and a number $\varepsilon \in (0,1].$}
		\ENSURE{Number $\mu$ such that 
			$\mathbb{E}[X_{\mathcal{G}^{(0)}}(\source,\sink)] - \varepsilon < \mu \leq  \mathbb{E}[X_{\mathcal{G}^{(0)}}(\source,\sink)]$}
		
		\medskip
		
		\STATE{Let $H=(V,E,w)$ be the weighted graph obtained from the underlying graph $G$ by assigning
		to every edge $e \in E$ the weight $w(e) = \frac{1}{p_e}$}
		
		\STATE{Compute the minimum weight $w^*$ of an $\source$-$\sink$ path in $H$}
		
		\STATE{Let $\tau = w^* \left( \ln \frac{w^*}{\varepsilon} + 1 \right)$}
		
		\STATE{Compute a tree decomposition $T$ of $(G,\source,\sink)$}\label{line:dynamDecomposition}
		
		\STATE{\textsc{Compute SP probabilities}($\mathcal{G}^{(0)}, \source, \sink_, T, \tau$)}
		
		\medskip
		
		\RETURN{$\sum_{i=1}^{\tau} \Pr[X_{\mathcal{G}^{(0)}}(\source,\sink) \geq i]$}
	\end{algorithmic}
\end{algorithm}

The following theorem follows from the proof of Theorem~\ref{th:generalFPTAS} and Theorem~\ref{th:spProbAlg}.

\begin{theorem}
	Algorithm~\ref{alg:FPTAS_SP} is correct and works in time 
	$O\Big( m \cdot n^{2c+2} \ln^2 \frac{n}{\varepsilon}\Big)$, where $n$ and $m$ are the number of vertices and
	the number of edges in the underlying graph, respectively.
\end{theorem}

\subsection{The FPRAS for general graphs in the memory-$k$ model, $k\geq 0$}\label{sec:FPRAS}

In this section, we present our FPRAS for \Po\ in the memory-$k$ model, for every $k\geq 0$, under the 
assumption that every edge appearance probability is lower bounded by $\frac{1}{n^c}$ for some $c\geq 1$.

\begin{lemma}\label{lem:fpras-1}
Let $\mathcal{G}^{(k)} = (G=(V,E), \{p_e(H_e^{(k)}) \} )$ be a memory-$k$ stochastic temporal graph,
and let $Z$ be a geometric random variable with success probability $\frac{1}{n^c}$. Then we have
\begin{enumerate}
	\item $\Pr\left[X(\source,\sink) \geq a\right] \ \leq \ \Pr\left[(n-1)Z \geq a \right]$, for every $a \in \mathbb{N}$, and
	\item $\mathbb{E}[X(\source,\sink)] \ \leq \ \mathbb{E}[(n-1)Z] \ \leq \ n^{c+1}$.
\end{enumerate}
\end{lemma}
\begin{proof}
	Let $e \in E$ be an arbitrary edge in the underlying graph and let $t \in \mathbb{N}$ be an arbitrary time point.
	Let $X_e^t$ be the random variable equal to the shortest time required to move from one of the end vertices of $e$ to the
	other starting at time $t$. Then for any $a \in \mathbb{N}$ we have that	
	\begin{eqnarray*}
			\Pr[X_e^t \geq a] & =& \prod_{i=0}^{a-2} \Pr[e \text{ does not appear at time point } t+i]  \\
			                  & =& \prod_{i=0}^{a-2} \left(1-\Pr[e \text{ appears at time point } t+i] \right) \\ 
			 & \leq& \left( 1- \frac{1}{n^c} \right)^{k-1} = \Pr[Z \geq a],
	\end{eqnarray*}
	where the inequality follows from our assumed lower bound on edge appearance probabilities.
	Now, since an edge and time point were picked arbitrarily and any $s$-$y$ foremost journey traverses
	at most $n-1$ edges, we conclude (1), that is, 
	$\Pr\left[X(\source,\sink) \geq a\right] \leq \Pr\left[(n-1)Z \geq a \right]$, for every $a \in \mathbb{N}$. Now, using (1) we establish (2) as follows:
\begin{equation*}
	 	\mathbb{E}[X(\source,\sink)] \ = \ \sum_{i=1}^{\infty} \Pr[X(\source,\sink) \geq i] \ \leq \ 
	 	 \sum_{i=1}^{\infty} Pr\left[nZ \geq a \right] \ = \ \mathbb{E}[(n-1)Z] \ \leq \ n^{c+1}. 
\end{equation*}
\end{proof}

In the following theorem we provide our FPRAS for \Po.

\begin{theorem}
\label{fpras-thm}
Let $\varepsilon\in (0,1)$ and let $\mathcal{G}^{(k)}$ be a memory-$k$ stochastic temporal graph 
with two designated vertices $\source,\sink$. 
Furthermore let every edge appearance probability be at least $\frac{1}{n^c}$ for some $c\geq 1$. 
Then \Po\ admits an FPRAS which runs in $O\left( m\frac{n^{5c+8}}{\varepsilon^4}  \cdot \log(\frac{n}{\varepsilon}) \right)$ time 
with probability of success at least $1-\frac{2}{n}$.
\end{theorem}

\begin{proof}
Let $\mathcal{G}^{(k)}$ be a stochastic temporal graph with two designated vertices $\source,\sink$. 
Furthermore let $X$, as before, be the arrival time of a foremost $\source$-$\sink$ journey. 
We will estimate the expectation $\mathbb{E}(X)$ via an unbiased estimator approach. 
We perform $r$ times independently the following experiment \textsc{Exp}; 
for now let us assume an arbitrary value for $r$, to be chosen precisely later. 

\begin{algorithm}[ht]
\caption*{\textbf{Experiment} \textsc{Exp}}
\label{experiment}
\begin{algorithmic}[1]
\REQUIRE{A stochastic temporal graph $\mathcal{G}^{(k)}$ on an underlying graph $G$ with $n$ vertices and $m$~edges and two designated vertices $\source,\sink$ of $G$}

\medskip

\STATE{Starting at time $t=0$, let $\mathcal{G}^{(k)}$ evolve until time $t' = r n^{c+2}$; 
the resulting temporal graph has at most $t'm$ time-edges}

\STATE{Run the foremost $\source$-$\sink$ journey algorithm of~\cite{akridaGMS17} in this temporal graph}

\medskip

\RETURN{the arrival time of the computed foremost journey}
\end{algorithmic}
\end{algorithm}

The probability  that \textsc{Exp} fails to connect $\source$ to $\sink$ via a journey is equal to the probability that $\source$ is not connected to $\sink$ until time $t'$. 
Therefore, Lemma~\ref{lem:fpras-1} implies that the time to connect $\source$ to $\sink$ exceeds the expectation $\mathbb{E}(X)$ of $X$ 
by a multiplicative factor of at least~$rn$. 
By Markov's inequality, this probability of failure is at most $\frac{1}{rn}$.
For now, we proceed the analysis of the algorithm assuming that 
all experiments succeed, and we will take the probability of failure of some experiment(s) into 
account later on.

Let $X_i$ be the random variable returned by the $i$th execution of the experiment \textsc{Exp}, 
where $i=1,2,\ldots,r$, and let $\widetilde{X} = \frac{1}{r}(X_1+\ldots+X_r)$ be the estimator for $X$. 
Note that $\mathbb{E}[\widetilde{X}] = \mathbb{E}[X]$, meaning that $\widetilde{X}$ is an unbiased estimator for $X$. 
Thus, it follows by Chebyshev's inequality that, for every $\varepsilon\in (0,1)$:
\begin{equation}
\Pr\left[  | \widetilde{X}  - \mathbb{E} [X] | \geq \varepsilon \mathbb{E}[     X   ]  \right]   \leq   \left( \frac{\sigma(\widetilde{X}) }{ \varepsilon \mathbb{E}[     X   ] } \right)^2. \label{eq:fpras2}
\end{equation}
It holds that $\sigma(\widetilde{X}) = \frac{\sigma(X)}{\sqrt{r}}$ (see~\cite[p.~297]{Vazirani}), hence~(\ref{eq:fpras2}) becomes:
\begin{equation}
\Pr\left[  | \widetilde{X}  - \mathbb{E} [X] | \geq \varepsilon \mathbb{E}[X]  \right]   \leq   \frac{\sigma^2(X)}{\varepsilon^2 r \mathbb{E}^2[X]} \leq   \frac{\sigma^2(X) }{ \varepsilon^2 r }. \label{eq:fpras3}
\end{equation}
We will now show that $\sigma(X)$ is upper bounded by a polynomial in $n$. Indeed, we have:
\begin{equation}
\sigma^2(X) = \mathbb{E}[X^2] - \mathbb{E}^2[X] \leq \mathbb{E}[X^2] \text{~(since } \mathbb{E}[X] \geq 0 \text{).} \label{eq:fpras4}
\end{equation}
Now, consider any $\source$-$\sink$ journey in $\mathcal{G}^{(k)}$ and let $\ell$ be the number of its edges. Let $Y_1$ be the number of time steps needed until we cross the $1$st edge of this journey, starting at time $0$. Also let $Y_i$ be the number of time steps that are needed, after we cross the $(i-1)$th edge, until we cross the $i$th edge of this journey, $i=2,\ldots,\ell$. Then $X \leq Y_1 + \ldots + Y_\ell$, and thus
\begin{equation*}
	X^2 \ \leq \ ( Y_1 + \ldots + Y_\ell )^2  \ \leq \ (1^2+ \ldots + 1^2) \cdot (Y_1^2 + \ldots + Y_\ell^2 ) \ \leq \ n (Y_1^2 + \ldots + Y_\ell^2 ) .
\end{equation*}

Note that each $Y_i$ is dominated by a geometric random variable $Z$ with success probability~$\frac{1}{n^c}$. Therefore, it follows by the properties of geometric random variables that $\mathbb{E}(Y^{2}_{i}) \leq 2 n^{2c}$. Thus $\mathbb{E}[X^2] \leq  n (\mathbb{E}[Y_1^2] + \ldots + \mathbb{E}[Y_\ell^2] )  \leq  
2n^{2c+1}\ell \leq 2n^{2c+2}$.   
So,~(\ref{eq:fpras4}) becomes:
\[ \sigma^2(X) \leq 2n^{2c+2}. \]
Going back to~(\ref{eq:fpras3}), it becomes:
\[  \Pr\left[  | \widetilde{X}  - \mathbb{E} [X] | \geq \varepsilon \mathbb{E}[X]  \right]   \leq    
2\frac{n^{2c+2}}{\varepsilon^2 r }. \]
So, for a number $r = 2\frac{n^{2c+3}}{\varepsilon^2}$ of experiments, we get:
\[  \Pr\left[  \widetilde{X} \in (1-\varepsilon, 1+\varepsilon) \cdot \mathbb{E} [X] \right]   \geq    1 - \frac{1}{n} ,       \]
meaning that performing $r$ independent times the experiment \textsc{Exp} results in polynomial time in a solution $\widetilde{X}$ that is within a factor $(1\pm \varepsilon)$ of the optimal with probability 
at least $1 - \frac{1}{n}$. 
Let us call the probability that $\widetilde{X}$ is \emph{far} from the optimal solution ``probability of failure of the estimator''. Recall that there is also a chance of failure in our algorithm if any of the $r$ experiments fails. Therefore, the probability of failure of our FPRAS is:
\begin{eqnarray*}
	\Pr[ \text{failure of FPRAS}]  &\leq& \Pr[\text{failure of some \textsc{Exp}}] + \Pr[\text{failure of the estimator $\widetilde{X}$}] \\	
	& \leq& r \frac{1}{rn} + \frac{1}{n} \ = \ \frac{2}{n} \xrightarrow{n \rightarrow +\infty} 0.
\end{eqnarray*}
We execute the experiment \textsc{Exp} for $r = 2\frac{n^{2c+3}}{\varepsilon^2}$ times. 
Each execution of the experiment runs in total for $t'=2\frac{n^{3c+5}}{\varepsilon^2}$ time; during this time 
we get at most $t'm$ time-edges, i.e.~the algorithm of~\cite{akridaGMS17} 
runs in $O(t'm \log(t'm))=O(m\frac{n^{3c+5}}{\varepsilon^2}  \cdot \log(\frac{n}{\varepsilon}))$ time. 
Thus the total running time is $O(m\frac{n^{5c+8}}{\varepsilon^4}  \cdot \log(\frac{n}{\varepsilon}))$.
\end{proof}

\section{Computing the expected arrival time of a best policy}
\label{sec:best-policy}

In this section we investigate the computational complexity of our second problem, namely \Pt.

\subsection{A polynomial-time algorithm for the memoryless model}\label{sec:dynamic_prog_memoryless}

In this section we focus on the memoryless model and we derive a polynomial-time dynamic-programming algorithm for \Pt.
We define for every vertex $v$ the expected arrival time 
$h(v, \sink) \stackrel{\text{def}}{=} \mathbb{E}_{\mathcal{G}^{(0)}}[Y(v,\sink)]$ of the $v$-$\sink$ journey suggested to Alice by a best policy (i.e.~when Alice starts her journey at vertex $v$). 
For simplicity of presentation, throughout Section~\ref{sec:dynamic_prog_memoryless} we write 
$h(v)\stackrel{\text{def}}{=}h(v,\sink)$.

Assume for now that for all $v \in V$, the value $h(v)$ is given; let $v_1=\sink, v_2, \ldots, v_n$ be an ordering of vertices of $V$ in non-decreasing values of $h$ (ties broken arbitrarily), namely $h(v_1) \leq h(v_2) \leq \cdots \leq h(v_n)$. Clearly, $v_1=y$ and $h(v_1) = h(y) =0$. 

Let $s_t$ be the vertex that Alice occupied at time $t$ and recall that 
$\Gamma_{G_{t}} (v)$ is the neighborhood of vertex $v$ in the snapshot $G_t$, for all $v \in V$ and all $t \in \mathbb{N}$.
Notice that, the best strategy of Alice at time $t+1$ is to look at all neighboring vertices of $s_t$ in $G_{t+1}$ and find one with minimum $h$-value, namely a vertex $u \in \arg \min\{h(v): v \in \Gamma_{G_{t+1}}(s_t)\}$. If $h(u) \geq h(s_t)$, then Alice has no incentive to change vertex and thus $s_{t+1}=s_t$. Otherwise, if $h(u) < h(s_t)$, then $s_{t+1}=u$. 

Therefore, to find the best choice for Alice, it suffices to find the values $h(v), v \in V$. In view of the above, 
if Alice is on vertex $v_i$ at time 0 (i.e.~she is on the $i$-th best vertex in terms of closeness to $y$), 
she will move to the $j$-th best (with $j<i$) only if an edge appears between $v_i$ and $v_j$ in the next 
step, and no edge to a vertex better than $v_j$ appears (i.e.~no edge between $v_i$ and $v_{\ell}, ~1\leq\ell \leq j-1$). 
This happens with probability $Q_{i, j} = p_{\{v_i, v_j\}} \prod_{\ell=1}^{j-1} (1-p_{\{v_i, v_{\ell}\}})$, 
where $\{v_i, v_\ell\}$ denotes the (undirected) edge between $v_i$ and $v_\ell$. 
Additionally, with probability $Q_i = \prod_{\ell=1}^{i-1} (1-p_{\{v_i, v_\ell\}})$ no edge to a vertex better 
than $v_i$ will appear, in which case Alice will stay on $v_i$. 
Therefore $h(v_i)$ can be recurrently computed by $h(v_i) = \sum_{j=1}^{i-1} Q_{i,j} h(v_j) + Q_i h(v_i) +1$, or equivalently:
\begin{equation}
\notag h(v_i) = \frac{\sum_{j=1}^{i-1} Q_{i,j} h(v_j) +1}{1- Q_i}, 
\end{equation}
with initial condition $h(v_1) = 0$. Indeed, the above equation follows by observing that the expected 
length of the foremost journey to $y$ when Alice is on $v_i$ is equal to $1+h(v_1)$ with probability $Q_{i, 1}$ 
(which is the probability that an edge between $v_i$ and $v_1=y$ exists), plus $1+h(v_2)$ with 
probability $Q_{i,2}$ (which is the probability that and edge between $v_i$ and the second best vertex $v_2$ exists, but there is no edge between $v_i$ and $v_1$), and so on. 
In general, the above recurrence states that there is no incentive to visit vertices with larger index and 
also Alice will visit the smallest index vertex $v_j$ for which the edge $\{v_i, v_j\}$ is present (otherwise, if no such edge exists, she will stay on $v_i$). 
Using the above recurrence, we can compute all values of $h(v_i)$ by the following bottom-up dynamic 
programming algorithm\footnote{To avoid trivialities, we assume that the graph induced by the set of edges with non-zero probabilities is connected; in particular, the set of probabilities can be used to model any 
connected underlying graph~$G$. We also assume that the elements of the list $L$ can be accessed using their index, i.e.~$L_i$ is the $i$-th element of the list.}:

\begin{algorithm}[ht]
\caption{\Pt\ in memoryless stochastic temporal graphs}
\label{alg:Best-Policy-Memoryless}  
\begin{algorithmic}[1]
\REQUIRE{A stochastic temporal graph $\mathcal{G}^{(0)} = (G=(V,E), \{ p_e \} )$.}
\ENSURE{The values $\{h(v): v \in V\}$, stored in the ordered list $L$.}

\medskip

\STATE{Let $L$ be the empty list}
\STATE{Append $y$ to $L$; \ \  $h(y)\leftarrow 0$}
\FOR{$i=2$ to $n$}
    \STATE{$u \leftarrow \arg \min \left\{\frac{\sum_{j=1}^{i-1} p_{\{v, L_j\}} \prod_{\ell=1}^{j-1} \left(1-p_{\{v, L_\ell\}}\right) h(L_j) +1}{1-\prod_{\ell=1}^{i-1} \left(1-p_{\{v, L_\ell\}}\right)} : \ v \notin L, \ \prod_{\ell=1}^{i-1} (1-p_{\{v, L_\ell\}})<1 \right\}$}
    \vspace{0,2cm}
    \STATE{$h(u)\leftarrow\frac{\sum_{j=1}^{i-1} p_{\{u, L_j\}} \prod_{\ell=1}^{j-1} \left(1-p_{\{u, L_\ell\}}\right) h(L_j) +1}{1-\prod_{\ell=1}^{i-1} \left(1-p_{\{u, L_\ell\}}\right)}$}
    \vspace{0,2cm}
    \STATE{Append $u$ to $L$}
\ENDFOR

\medskip

\RETURN{$L$ and $h(v), v \in V$}
\end{algorithmic}
\end{algorithm}

Algorithm~\ref{alg:Best-Policy-Memoryless} can be efficiently implemented to run in $O(n^2)$ time and space.
This can be achieved by carefully storing intermediate sums and products in step 4; indeed, at step $i+1$ the only new term in the numerator is $p_{\{v, L_i\}} \prod_{\ell=1}^{i-1} (1-p_{\{v, L_\ell\}}) h(L_i)$, 
and in the product $\prod_{\ell=1}^{i} (1-p_{\{v, L_\ell\}})$ which appears in the denominator the only new factor is $1-p_{\{v, L_i\}}$. Concluding, we have the following theorem:

\begin{theorem}
\label{best-policy-memory-0-thm}
\Pt\ can be optimally computed in the memoryless model in $O(n^2)$ time and space.
\end{theorem}

\subsection{Hardness of computation for the memory-$k$ model, $k\geq3$}\label{sec:hardness_for_Pt}
We now show that \Pt\ is {\#P}-hard for memory-3 stochastic temporal graphs on directed acyclic graphs, and consequently also for memory $k\geq 3$. 

\begin{theorem}
	When the underlying graph is a Directed Acyclic Graph (DAG), it is {\#P}-hard to compute the expected arrival time of the best
	policy journey in the memory-3 model.
\end{theorem}

\begin{proof}
	We will provide a reduction from the counting problem \#PP2DNF which is
	known to be {\#P}-hard~\cite{PB83}. Recall that this problem takes as input a DNF formula 
	$\Phi =\bigvee_{(i,j)\in E}x_{i}y_{j}$ on the sets of variables $%
	X=\{x_{1},\ldots ,x_{n}\}$ and $Y=\{y_{1},\ldots ,y_{m}\}$,
	for some $E \subseteq [n] \times [m]$, and the task is
	to compute the number $\psi$ of truth assignments that satisfy $\Phi $.
	Similarly to our reduction in the proof of Theorem~\ref{thm:prob_1_hard}, we
	create a directed acyclic graph (DAG) $H$ as follows. First, $H$ has one
	vertex for each of the variables in $X\cup Y$; then we add two distinct
	vertices $\source,\sink$ and one other vertex $v$. For every vertex $x_{i}\in X$ and
	every vertex $y_{i}\in Y$ we add the directed edges $(\source,x_{i})$ and $%
	(y_{j},\sink)$. Furthermore we add the edge $(x_{i},y_{j})$ whenever $x_{i}y_{j}$
	is a clause in $\Phi $. Finally we add the edges $(\source,v)$ and $(v,\sink)$. The
	construction of $H$ is illustrated in Figure~\ref{hardness-memory-3-graph-fig}.
	
	\begin{figure}[ht]
		\centering
		\includegraphics[scale=0.6]{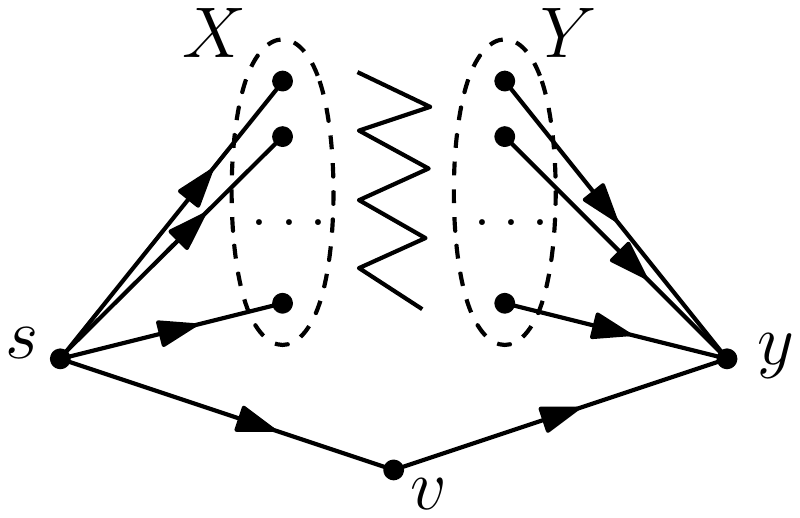}
		\caption{The construction of the DAG $H$.}
		\label{hardness-memory-3-graph-fig}
	\end{figure}

	Denote by $M=5\cdot 2^{n+m}$, and assume that $2^{n+m}\geq 3$ in order to
	avoid trivialities. All edges $(x_{i},y_{j})$ appear constantly in $H$, 
	i.e.~they appear at every time step $i\geq 1$ in a memoryless fashion with
	probability $1$. Both edges $(\source,v)$ and $(v,\sink)$ also appear in a memoryless
	fashion, each of them with probability~$\frac{2}{M}$ at every step $i\geq 1$. 
	Moreover, each of the edges $(\source,x_{i})$ and $(y_{j},\sink)$ appears at each
	step $i\geq 1$ according to the following table of memory $3$. This table
	has four columns and eight rows. Each column is labeled with the sequence of
	consecutive time steps $i-3,i-2,i-1$, and $i$. Each row corresponds to a
	different triple of appearances of each of the edges in $%
	\{(\source,x_{i}),(y_{j},\sink):x\in X,y\in Y\}$ at the time steps $i-3,i-2,i-1$ (here 
	$1$ means ``edge exists'' and $0$ means
	``edge does not exist''). At the end of
	each row there is a pair of numbers $(p,1-p)$ which denotes that, with the
	particular history of memory $3$, at time step $i$ the edge appears with
	probability $p$ and it does not appear with probability $1-p$. For
	simplicity of notation, in the column of time step $i$, we write
	``0'' and ``1'' to denote the entries $(0,1)$ and $(1,0)$, respectively.
	
	\begin{center}
		\begin{tabular}{ccc|c}
			$i-3$ & $i-2$ & $i-1$ & $i$ \\ \hline
			0 & 0 & 1 & 0 \\ 
			0 & 1 & 0 & ($\frac{1}{2}$,$\frac{1}{2}$) \\ 
			1 & 0 & 0 & 0 \\ 
			0 & 0 & 0 & 0 \\ 
			1 & 0 & 1 & 1 \\ 
			0 & 1 & 1 & 1 \\ 
			1 & 1 & 1 & 1 \\ 
			1 & 1 & 0 & 1%
		\end{tabular}
	\end{center}
	
	To complete the description of our memory-3 instance, we specify that, in
	the fictitious initialization snapshots $G_{-2},G_{-1},G_{0}$, each of the
	edges $(\source,x_{i})$ and $(y_{j},\sink)$ appears with probability $0$, $0$, and $1$, 
	respectively, i.e.~according to the first row of the above table.
	
	The intuition of this table for the edges $(\source,x_{i})$ and $(y_{j},\sink)$ is as
	follows. In the snapshot $G_{1}$, none of these edges appears (see the first
	line of the table). Then, to determine whether each of these edges appears
	at time step $2$ (see the second row of the table), we need to toss an
	unbiased coin which with probability $\frac{1}{2}$ outputs ``appear'' 
	and with probability $\frac{1}{2}$ outputs ``does not appear''. Once this coin has
	been tossed at time step $2$, the status of the edge does not change any
	more in any subsequent time step $i\geq 3$. That is, if one of the edges $%
	(\source,x_{i})$ and $(y_{j},\sink)$ appears (resp. does not appear) at time $2$, then
	it appears (resp. does not appear) at all times $i\geq 3$ too. This is easy
	to be verified by observing the rows $3$-$7$ of the table. Note that the
	last row of the table is included only for the sake of completeness, as it
	does not affect the appearance of any edge of $H$ at any time step $i$.
	
	Let $\ell $ be the expected $\source$-$\sink$ arrival time of the best policy in the
	memory-3 model. Note that, from the above construction of the temporal graph
	instance, each of the edges $(\source,x_{i})$ and $(y_{j},\sink)$ appears with
	probability $\frac{1}{2}$ at all steps $i\geq 2$, while it does not appear
	at any step $i\geq 2$ with probability $\frac{1}{2}$. Therefore, the
	probability that there exists a directed temporal path $(\source,x_{i},y_{j},\sink)$
	is equal to $g=\frac{\psi}{2^{n+m}}$, where $\psi$ is the number of satisfying
	truth assignments of the DNF formula $\Phi $. That is, with probability $%
	1-g $, there exists no such temporal path from $\source$ to $\sink$ with 3 edges
	through some vertices $x_{i}$ and $y_{j}$. Furthermore, the expected $\source$-$\sink$
	arrival time through the edges $(\source,v)$ and $(v,\sink)$ is equal to $\frac{M}{2}+%
	\frac{M}{2}=M$. Therefore, since with probability $1-g$ any policy (also the
	best one) needs to travel from $\source$ to $\sink$ through vertex $v$, it follows
	that $\ell \geq M(1-g)$.
	
	We now define the following policy: at time step $1$ do nothing and just
	wait for the outcome of the random coin tosses which occur at time step $2$.
	Subsequently, at time step $2$ do the following: if there exists a directed
	temporal path $(\source,x_{i},y_{j},\sink)$ then follow it, starting at time step $2$;
	otherwise follow the temporal path $(\source,v,\sink)$ which has an expected travel
	time $\frac{M}{2}+\frac{M}{2}=M$. The expected arrival time of this
	particular policy is equal to $1+3g+M (1-g)$, and thus it follows that $%
	\ell \leq 1+3g+M (1-g)$. Summarizing, we have:%
	\begin{eqnarray*}
		M(1-g) &\leq &\ell\ \ \leq\ \ 1+3g+M (1-g) \Leftrightarrow  \\
5\cdot 2^{n+m}-5\psi &\leq &\ell \ \ \leq \ \ 5\cdot 2^{n+m}-5\psi+3\frac{\psi}{2^{n+m}}+1.
	\end{eqnarray*}%
	The first inequality can be written as $2^{n+m}-\frac{\ell }{5}\leq \psi$,
	while the second one can be written as $\left( 1-\frac{3}{5\cdot 2^{n+m}}%
	\right) \psi\leq 2^{n+m}-\frac{\ell }{5}+\frac{1}{5}$. Therefore: 
	\begin{equation*}
		2^{n+m}-\frac{\ell }{5} \leq \psi\leq \left( 1+\frac{3}{5\cdot 2^{n+m}-3}\right) 
		\left( 2^{n+m}-\frac{\ell }{5}+\frac{1}{5}\right) \leq 2^{n+m}-\frac{\ell }{5}+\frac{1}{5}+\frac{3}{4},
	\end{equation*}%
	and thus%
	\begin{equation}
	2^{n+m}-\frac{\ell }{5}\leq \psi\leq 0.95+2^{n+m}-\frac{\ell }{5}.
	\label{eq-for-k}
	\end{equation}%
	Thus, knowing the expected value $\ell $ for the best policy we can derive
	the exact integer value for $\psi$ in the counting problem \#PP2DNF. This completes the
	{\#P}-hardness reduction. 
\end{proof}

\subsection{An exact algorithm for the memory-$k$ model, $k\geq 1$} \label{sec:bestpolicy}

In this section we present a doubly exponential-time exact algorithm for computing the best policy for Alice in the memory-$k$ model, where $k\geq 1$. 
We first give a Markov Decision Process (MDP) formulation of our problem under the memory-$k$ model that will be useful 
for the presentation of our results within the general MDP framework.

\subsubsection{An MDP formulation} \label{sec:mdpformulation}

We follow the notation from chapter 5.4 of \cite{Norris} to give an MDP formulation. Let $k \geq 1$ be a fixed integer corresponding to the memory of the model. We denote by ${\cal I} = V \times (2^{G})^k$ the set of \emph{states}, where $2^G$ denotes the set of subgraphs of the underlying graph $G$. In particular, each state $(v, H^{(k)}) \in {\cal I}$ consists of a vertex $v$ which corresponds to the vertex Alice is on and a sequence of $k$ graphs $H^{(k)}$ corresponding to the $k$ most recent snapshots. 
For any $t \geq 0$, we will say that $H^{(k)}_t \stackrel{\text{def}}{=} \left(G_{t-k+1}, G_{t-k+2}, \ldots, G_{t-1}, G_t\right)$ occurred at time $t$, if the snapshots at times $t-k+1, t-k+2, \ldots, t-1, t$ are $G_{t-k+1}, G_{t-k+2}, \ldots, G_{t-1}, G_t$, respectively. The set of \emph{actions} for Alice is the set ${\cal A} = V$. A \emph{stationary policy} for Alice is a function $f:{\cal I} \to {\cal A}$ and determines a probability law $\Pr^f$ for a Markov chain $(X_t)_{t \geq 0}$ with values in ${\cal I}$ as follows:
\begin{description}

\item[(i)] Assuming that at time 0 Alice starts from vertex $s$ and the initial sequence of $k$ snapshots is $H^{(k)}_0 = \left(G_{-k+1}, G_{-k+2}, \ldots, G_{-1}, G_0\right)$, the initial distribution is given by $\Pr^f\left[X_0=(s, H^{(k)}_0) \right] = 1$, and $\Pr^f\left[X_0=(v, H^{(k)}) \right] = 0$ if $v \neq s$ or $H^{(k)} \neq H^{(k)}_0$.

\item[(ii)] For any $t \geq 0$,
\begin{eqnarray}
&& {\Pr}^f\left(X_{t+1}=(v_{t+1}, H^{(k)}_{t+1}) | X_t=(v_t, H^{(k)}_t) \right) \nonumber \\
&& \quad = \left\{ 
\begin{array}{ll}
	\Pr[\textrm{$G_{t+1}$ occurs at $t+1$}| \textrm{$H^{(k)}_t$ occurred at $t$}] & \textrm{if $f(v_t, H^{(k)}_{t+1})=v_{t+1}$} \\
	0 & \textrm{if $f(v_t, H^{(k)}_{t+1}) \neq v_{t+1}$}
\end{array}
\right.
\end{eqnarray} 
\end{description}

Without loss of generality, we will assume that every policy $f$ is \emph{legitimate} in the sense that the following conditions hold:
\begin{description}
\item[A.] $f(v_t, H^{(k)}_{t+1}) = v_{t+1}$ only if $(v_t, v_{t+1}) \in E(G_{t+1})$, i.e.~Alice may visit $v_{t+1}$ in the next step only if $G_{t+1}$ has an edge that connects $v_t$ (the vertex she is currently on) and $v_{t+1}$ (the vertex she wants to go to). 

\item[B.] Recalling that the goal of Alice is to reach $y$, we assume that $f(y, H^{(k)})=y$, for any $H^{(k)}$, i.e.~Alice will never leave her target vertex once she reaches it.
\end{description}
For simplicity, we will denote by $a_t$ Alice's $t$-th action (vertex choice). In particular, $a_0 = s$ and inductively $a_{t+1} = f(a_t, H^{(k)}_{t+1})$, for any $t \geq 0$. Furthermore, let $\mu(G_{t+1}|H^{(k)}_t) \stackrel{\text{def}}{=} \Pr[\textrm{$G_{t+1}$ occurs at $t+1$}| \textrm{$H^{(k)}_t$ occurred at $t$}]$.

To complete the specification of the Markov Decision Process, we assume that constant cost $c((v, H^{(k)}), a)=1$ is incurred when action $a$ is chosen in state $(v, H^{(k)})$ with $v \neq y$, otherwise $c((y, H^{(k)}), a)=0$. Therefore, to every legitimate policy $f$ we can associate an expected total cost starting from state $(a_0, H^{(k)}_0)$, given by $h^f(y, y, H^{(k)}) = 0$, for any $H^{(k)}$ and, for any $a_0 \neq y$ and any $H^{(k)}_0$,
\begin{eqnarray}
\hspace{-0,5cm}
 h^f(a_0, y, H^{(k)}_0) &=& \mathbb{E}^f\left[ \sum_{t=0}^\infty c((a_t, H^{(k)}_t), a_{t+1})\right] = \mathbb{E}^f\left[ \sum_{t=0}^\infty c((a_t, H^{(k)}_t), f(a_t, H^{(k)}_{t+1}))\right] \label{eq:start} \\
&=& 1 + \mathbb{E}^f\left[ \sum_{G_1} \mu(G_1|H^{(k)}_0) \sum_{t=1}^\infty c((a_t, H^{(k)}_t), a_{t+1}) \right] \label{eq:condition1} \\
&=& 1 + \sum_{G_1} \mu(G_1|H^{(k)}_0) h^f(a_1, y, H^{(k)}_1) \label{eq:condition1.1}.
\end{eqnarray}
To be more clear, the expectations in equation (\ref{eq:start}) are over random variables 
$G_1, G_2, \ldots$, while the expectation in equation (\ref{eq:condition1}) 
is over $G_2, G_3, \ldots$. Furthermore, equation (\ref{eq:condition1}) follows by conditioning 
on $G_1$ and equation (\ref{eq:condition1.1}) follows by observing that, 
by symmetry, $\mathbb{E}^f\left[\sum_{t=1}^\infty c((a_t, H^{(k)}_t), a_{t+1})\right]$ equals 
the expected total cost starting from $(a_1, H^{(k)}_1)$.

\begin{observation}
\emph{Any} policy $f$ guiding Alice from $s$ to $y$ must satisfy recurrence (\ref{eq:condition1.1}), with initial condition $h^f(y, y, H^{(k)}) = 0$, for every $H^{(k)}$.
\end{observation}

Our objective is to find a policy that minimizes the expected total cost $h^f(a_0, y, H^{(k)}_0)$. In particular, this policy will have the value $h^*(a_0, y, H^{(k)}_0) = \inf_f h^f(a_0, y, H^{(k)}_0)$ which will be equal to the expected arrival time of a journey suggested to Alice by an optimal policy. In fact, without loss of generality we will assume that the $h^*$-values of an optimal policy satisfy $h^*(a_0, y, H^{(k)}_0) = \inf_f h^f(a_0, y, H^{(k)}_0)$, for all $a_0 \in V$ and all $H^{(k)}_0 = \left(G_{-k+1}, G_{-k+2}, \ldots, G_{-1}, G_0\right)$, such that $G_i \subseteq G$, for all $-k+1 \leq i \leq 0$.

\subsubsection{A doubly exponential-time algorithm} \label{sec:doublyexp}

We now provide our doubly exponential-time algorithm for \Pt\ in the memory-$k$ model, where $k\geq 1$. 
In order to simplify the notation and presentation of this section, we only provide the proof of the algorithm for the special case $k=1$; the analysis for arbitrary $k\geq 1$ carries then easily over, 
as we discuss at the end of the section.

Following the notation of Section~\ref{sec:mdpformulation} for memory-1, we denote by $\mu(G''|G')$ the probability that the next snapshot is $G''$, given that the current snapshot is $G'$. Furthermore, for $a_0 \in V$, let $h(a_0, y, G_0)$ be the expected arrival time of a journey from $a_0$ to $y$ suggested to Alice by an optimal policy, given that the starting graph instance is equal to $G_0$.

We define the following policy $\pi$: For any time step $t \geq 0$, if at time $t$ Alice was on a vertex $a_t$ and at time $t+1$ the graph instance is $G_{t+1}$, then at time $t+1$ she will move to a vertex $u \in \Gamma_{G_{t+1}}[a_t]$ that has minimum $h(u, y, G_{t+1})$, that is, 
\begin{equation}
\pi(a_t, G_{t+1}) \stackrel{\text{def}}{=} a_{t+1} \in \arg\min\left\{h(u, y, G_{t+1}): u \in \Gamma_{G_{t+1}}[a_t] \right\}.
\end{equation}
By part (ii) of Theorem 5.4.3 of \cite{Norris}, we have the following:

\begin{lemma} \label{lemma:optimalpi}
Policy $\pi$ is optimal.
\end{lemma}

Notice that in the definition of $\pi$, we assumed that the $h$-values are given. Therefore, to determine $\pi$ we need to compute $h(a_0, y, G_0)$, for every $a_0 \in V$ and $G_0 \subseteq G$. We start by rewriting recurrence (\ref{eq:condition1.1}) for policy $\pi$:
\begin{equation} \label{eq:recurrencememory1}
h^{\pi}(a_0, y, G_0) = 1 + \sum_{G_1} \mu(G_1|G_0) h^{\pi}(\pi(a_0, G_1), y, G_1).
\end{equation} 
Since $\pi$ is optimal, the left hand side of the above equation is equal to $h(a_0, y, G_0)$. Furthermore, by definition of $\pi(a_0, G_1)$, 
\begin{equation}
h^{\pi}(\pi(a_0, G_1), y, G_1) = \min\left\{h(u, y, G_1): u \in \Gamma_{G_1}[a_0] \right\}.
\end{equation}
Therefore, recurrence (\ref{eq:recurrencememory1}) becomes
\begin{equation} \label{eq:recurrencepi}
h(a_0, y, G_0) = 1 + \sum_{G_1} \mu(G_1|G_0) \min\left\{h(u, y, G_1): u \in \Gamma_{G_1}[a_0] \right\}.
\end{equation}

Suppose that we know an ordering of the triplets $(a_0, y, G_0), a_0 \in V, G_0 \subseteq G$, in increasing values of $h(a_0,y,G_0)$, breaking ties arbitrarily. Notice that these are $n2^{m} \stackrel{\text{def}}{=} N$ values, where $m=|E|$ is the number of edges of $G$. Then the minimum in recurrence (\ref{eq:recurrencepi}) can be replaced with the corresponding $h$-value, which is completely determined by the graph $G_1$ and the vertex $a_0$. Doing this for all different vertices $a_0$ and graphs $G_0$, we get a linear system with $N$ equations coming from (\ref{eq:recurrencepi}) and as many variables (the $h$-values). To this system, we then add the initial conditions $h(y,y, G_0)=0$, for all $G_0 \subseteq G$. This can be solved in $O(N^3)$ time. 

Notice however, that for the above approach to work, we need an ordering of the triplets $(a_0, y, G_0)$ in increasing values of $h(a_0,y,G_0)$. We can therefore have the following brute-force algorithm: For each of the (at most) $N!$ orderings of the triplets $(u_0, y, G_0), a_0 \in V, G_0 \subseteq G$, solve the linear system derived by the recurrence (\ref{eq:recurrencepi}) as described above, assuming the ordering is ``correct'', namely it corresponds to an ordering in increasing values of $h(a_0,y,G_0)$. Then check if the ordering we get from the solution to that system is the same as the one we assumed. If not, then consider a different ordering. 

Notice however the following issue with the above brute-force algorithm: suppose the correct ordering $\sigma^*$ is considered and we construct the corresponding linear system of equations (call it $\Sigma$) based on (\ref{eq:recurrencepi}). Clearly, $\Sigma$ has at least one solution, but what happens if there are more than one solutions, some of which giving an ordering that is not consistent with $\sigma^*$? Can we find the correct solution among all other solutions of $\Sigma$? To circumvent this problem, we replace the linear system of equalities with a linear system of inequalities (constraints). To this end, let $u_{a_0, G_1}^{\sigma^*}$ be the vertex such that $u_{a_0, G_1}^{\sigma^*} \in \Gamma_{G_1}[a_0]$ and the triplet $(u_{a_0, G_1}^{\sigma^*}, y, G_1)$ appears in $\sigma^*$ before all triplets $(u, y, G_1)$, for all $u \in \Gamma_{G_1}[a_0] \setminus \{u_{a_0, G_1}^{\sigma^*}\}$. We want to find \emph{any} solution satisfying the following constraints: 
\begin{eqnarray}
&& h'(a_0, y, G_0) = 1 + \sum_{G_1} \mu(G_1|G_0) h'(u_{a_0, G_1}^{\sigma^*}, y, G_1) \label{constraint:recurrence}\\
&& h'(u_{a_0, G_1}^{\sigma^*}, y, G_1) \leq h'(u, y, G_1) , \quad \textrm{$\forall u \in \Gamma_{G_1}[a_0], a_0 \in V\backslash \{y\}, G_0 \subseteq G$} \label{constraint:min} \\
&& h'(a_0, y, G_0) \geq 0 \quad \textrm{$\forall a_0 \in V, G_0 \subseteq G$} \label{constraint:nonnegative}\\
&& h'(y, y, G_0) = 0 \quad \textrm{$\forall G_0 \subseteq G$}. \label{constraint:zero}
\end{eqnarray}  

By definition of $\sigma^*$, the above set of constraints has at least one solution, namely the one corresponding to the $h$-values of an optimal policy. In Theorem~\ref{thm:uniquesolution} we prove that this is the only feasible solution. For the proof, we also need the following Theorem from \cite{Norris}, which we restate here in our notation for convenience:

\begin{theorem}[Policy increment, Theorem 5.4.4, \cite{Norris}] \label{thm:policyincrement}
Given one stationary policy $f$, let $\theta f$ denote the policy that, for every $a_0, G_0$ minimizes $\sum_{G_1} \mu(G_1|G_0) h^f((\theta f)(a_0, G_1), y, G_1)$. Then, for all $a_0, G_0$
\begin{equation}
\lim_{k \to \infty} h^{\theta^k f}(a_0, y, G_0) = h(a_0, y, G_0),
\end{equation}
provided $\mathbb{E}_{(a_0, G_0)} h^f(a_n, y, G_n) \to 0$ as $n \to \infty$.
\end{theorem}
In the above, the notation $\theta^k f$ means the application of policy increment $k$ times. Furthermore, in the expectation $\mathbb{E}_{(a_0, G_0)} h^f(a_n, y, G_n)$, the state $(a_n, G_n)$ is a random variable and its distribution is determined by an optimal policy, given that we start at $(a_0, G_0)$. We note that, the condition of the Theorem holds in our case by transience of the underlying Markov chain (i.e.~once Alice reaches $y$ she does not leave and no further cost is incurred after that).
 
We now prove the following:
\begin{theorem} \label{thm:uniquesolution}
Let $\sigma^*$ be an ordering of the triplets $(a_0, y, G_0), a_0 \in V, G_0 \subseteq G$, in increasing order of $h$-values. Let $(h'^*(a_0, y, G_0): a_0 \in V, G_0 \subseteq G)$ be any feasible solution to the set of constraints (\ref{constraint:recurrence}) to (\ref{constraint:zero}). Then $h'^*(a_0, y, G_0) = h(a_0, y, G_0)$, for all $a_0, G_0$.
\end{theorem}
\begin{proof}
We define the following policy $\pi^*$ (similar to the definition of $\pi$ earlier, but using the $h'^*$-values instead of the $h$-values): For any time step $t \geq 0$, if at time $t$ Alice was on a vertex $a_t$ and at time $t+1$ the graph instance is $G_{t+1}$, then at time $t+1$ she will move to a vertex $u \in \Gamma_{G_{t+1}}[a_t]$ that has minimum $h'^*(u, y, G_{t+1})$, that is, 
\begin{equation}
\pi^*(a_t, G_{t+1}) \stackrel{\text{def}}{=} a_{t+1} \in \arg\min\left\{h'^*(u, y, G_{t+1}): u \in \Gamma_{G_{t+1}}[a_t] \right\}.
\end{equation}
Since the $h'^*$-values satisfy constraint (\ref{constraint:recurrence}) and also, by constraint (\ref{constraint:min}), $u_{a_0, G_1}^{\sigma^*} \in \arg\min\{ h'^*(u, y, G_1): u \in \Gamma_{G_1}[a_0]\}$, the expected arrival time of a journey from $a_0$ to $y$ when Alice follows policy $\pi^*$ is equal to $h'^*(a_0, y, G_0)$, for all $a_0 \in V$ and $G_0 \subseteq G$. 

Observe that, if $\pi^*$ is not optimal, then we can successively apply policy increments $\theta^k \pi^*, k \to \infty$ as in Theorem~\ref{thm:policyincrement}, to eventually reach optimality. Notice also that the sum in the definition of policy increment for $\pi^*$ can be written as
\begin{equation}
\sum_{G_1} \mu(G_1|G_0) h^{\pi^*}((\theta \pi^*)(a_0, G_1), y, G_1) = \sum_{G_1} \mu(G_1|G_0) h'^*((\theta \pi^*)(a_0, G_1), y, G_1).
\end{equation}
Therefore, by definition, $\pi^*$ itself is a policy that minimizes the above sum, and so we can take $\theta \pi^* = \pi^*$. Consequently, no improvement by increment is possible, implying that $\pi^*$ is optimal. In particular, $(h'^*(a_0, y, G_0): a_0 \in V, G_0 \subseteq G)$ is the same as $(h(a_0, y, G_0): a_0 \in V, G_0 \subseteq G)$, and the proof is completed.
\end{proof}

The set of constraints (\ref{constraint:recurrence}) to (\ref{constraint:zero}) has $N=n2^m$ variables, namely $\{h'(a_0, y, G_0): a_0 \in V, G_0 \subseteq G\}$. 
Furthermore, there are $(n-1)2^m$ constraints of the form (\ref{constraint:recurrence}), at most $n^2 2^{m}$ constraints of the form (\ref{constraint:min}) and $n2^{m}$ non-negativity and initialization constraints, i.e.~$O(nN)$ constraints in total. 
Therefore, Vaydia's algorithm for linear programming \cite{V89} can find an optimum solution in $O((nN)^{2.5})$ time. Since we need to solve this set of constraints for every possible ordering of the $N$ different triplets $(a_0, y, G_0)$, our brute-force approach runs in $O(N! \ (nN)^{2.5})=O(N^N)$ time.

The above analysis for the memory-1 model directly carries over to the memory-$k$ model, for any $k\geq 1$. 
Indeed, the correctness proof can be slightly modified by replacing everywhere 
the subgraphs $G_t$ of $G$ by the length-$k$ histories $H^{(k)}_t$, respectively. 
Furthermore, the running time analysis carries over to the case of an arbitrary $k\geq 1$ by 
replacing $N=n2^{m}$ by $N'=n2^{km}$. 
Summarizing, we obtain the following theorem.

\begin{theorem}
\label{thm-doubly-exponential-k}
Let $k\geq 1$ and $\mathcal{G}^{(k)}$ be a stochastic temporal graph, 
where the underlying graph $G$ has $n$ vertices and $m$ edges. 
Then \Pt\ can be solved on $\mathcal{G}^{(k)}$ in $O(2^{(kmn+n\log n)\cdot 2^{km}})$ time.
\end{theorem}

\begin{remark}
It is easy to see that the running time of the above brute-force algorithm is dominated by the number of different orderings $N'!$, and thus we have a doubly exponential algorithm (recall that $N'=n2^{km}$). A different approach that can potentially lead to a faster algorithm is to start from an arbitrary initial policy and successively apply policy increments as in Theorem~\ref{thm:policyincrement}. 
Even though the convergence analysis of such an approach is non-trivial, one could use it to find the optimal ordering $\sigma^*$ fast\footnote{This is possible if after a relatively small (say polynomial in $N'$) number of steps the ordering does not change.} and then use $\sigma^*$ to find the unique solution to the set of constraints (\ref{constraint:recurrence})-(\ref{constraint:zero}).  
\end{remark}

\end{document}